\documentclass{article}
\usepackage[utf8]{inputenc}
\usepackage{amsmath,amssymb,amsfonts,bm}
\usepackage{appendix}
\usepackage{graphicx}
\usepackage{epstopdf}
\usepackage{mathrsfs}
\usepackage{amsthm}
\usepackage[margin=25mm]{geometry}
\usepackage{fullpage}
\usepackage{galois,booktabs,threeparttable,tabularx}
\usepackage[misc]{ifsym}
\usepackage{authblk}
\usepackage{subcaption}
\usepackage{cite}
\usepackage[capitalise]{cleveref}
\usepackage{enumitem}
\usepackage[ruled,linesnumbered]{algorithm2e}
\usepackage{etoolbox}
\usepackage{geometry}

\title{An Efficient Two-Stage SPARC Decoder for Massive MIMO Unsourced Random Access}
\author{Juntao You}
\author{Wenjie Wang\thanks{Corresponding author}}
\author{Shansuo Liang}
\author{Wei Han}
\author{Bo Bai}
\affil{Theory Lab, Central Research Institute, 2012 Labs, Huawei Technologies Co. Ltd., China \\ \{youjuntao, wang.wenjie, liang.shansuo, harvey.hanwei, baibo8\}@huawei.com }

\makeatletter
\patchcmd{\@algocf@start}
{-1.5em}
{0pt}
{}{}
\makeatother

\def\lV{\left\lVert}
\def\rV{\right\lVert}
\def\lv{\left\lvert}
\def\rv{\right\lvert}

\def\l{\langle}
\def\r{\rangle}

\def\H{\mathcal H}

\def\A{\mathcal A}

\def\Y{\mathcal Y}

\def\Y{\mathcal Y}

\newtheorem{theorem}{Theorem}

\newtheorem{lemma}{Lemma}

\newtheorem{definition}{Definition}

\newtheorem{proposition}{Proposition}
\newtheorem{remark}{Remark}
\newtheorem{example}{Example}
\graphicspath{{figures/}}

\begin{document}

\maketitle

\begin{abstract}
       In this paper, we study a concatenate coding scheme based on sparse regression code~(SPARC) and tree code for unsourced random access in massive multiple-input and multiple-output systems. Our focus is concentrated on efficient decoding for the inner SPARC with practical concerns. A two-stage method is proposed to achieve near-optimal performance while maintaining low computational complexity. Specifically, a one-step thresholding-based algorithm is first used for reducing large dimensions of the SPARC decoding, after which a relaxed maximum-likelihood estimator is employed for refinement. Adequate simulation results are provided to validate the near-optimal performance and the low computational complexity. Besides, for covariance-based sparse recovery method, theoretical analyses are given to characterize the upper bound of the number of active users supported when convex relaxation is considered, and the probability of successful dimension reduction by the one-step thresholding-based algorithm.
\end{abstract}


\section{Introduction}

Massive machine type communication~(mMTC) is one of the primary advancements in the internet of things~(IoT) applications~\cite{carvalho2017random,cirik2019toward,zhan2018massive}. In a typical mMTC scenario, a massive amount of potential user equipment~(UE) sends sporadic data to a base station~(BS) and the BS needs to decode the message from active users~\cite{polyanskiy2017a,fengler2019sparcs,fengler2021non,mohammadkarimi2022massive}. In general, the UE in mMTC is a low-cost and battery-powered sensing device and has limited computational resources. Within a certain time slot, only a portion of potential UE would be active and send short packets (around $10^2$ bits/packet) to the BS and the UE in mMTC is usually latency insensitive. Overall, in mMTC, the BS could support a massive amount of UE with decoding algorithm of slightly higher computational complexity.

Most cellular standards, such as 3G, 4G-LTE and 5G-NR, utilize grant-based schemes for transmission~\cite{cirik2019toward,zhan2018massive}. In the grant-based schemes, each active user is associated with an unique preamble to perform random access procedures, where a user could identify itself and ask for transmission resources. The BS proactively schedules transmission resources for each identified user. However, the sporadic communication patterns and the massive amount of potential users in mMTC cause huge scheduling overhead, rendering such grant-based scheme inefficient and impractical. Therefore, grant-free based protocols have attracted significant attentions recently, where UE sends data directly to the BS without approval. With carefully designed collision resolution mechanisms, grant-free schemes require small scheduling overhead and are highly efficient, which consequently seems favorable for mMTC~\cite{cirik2019toward,bockelmann2016massive}.

Recently, a new grant-free random access paradigm called ``unsourced random access'' (URA) was proposed for mMTC~\cite{polyanskiy2017a,mohammadkarimi2022massive}. In URA, potential users employ the same codebook and transmit messages without revealing their identities. The receiver aims to recover a list of permuted messages without knowing identities of active users~\cite{fengler2019sparcs}. Particularly, we consider an uplink multi-user grant-free system over a block-fading wireless communication channel, where the channel coefficients remain constant over coherence blocks of $D$ channel uses, and change randomly from block to block according to a stationary ergodic process~\cite{tse2005fundamentals}. Within a coherent block, only $K\ll K_{total}$ users are active.\footnote{Here, as in~\cite{polyanskiy2017a,fengler2021non,fengler2019sparcs}, we consider that users are perfectly synchronized in both time and frequency domain.} Denote $\bm{a}_k=[a_{k,1},a_{k,2},\cdots,a_{k,D}]^T\in \mathbb{C}^{D}$ as the transmitted signal of user $k$. Assuming that BS is equipped with $M$ antennas, the received signal is given by
\begin{equation}\label{problem:primal}
       \bm{y}_i=\sum_{k=1}^{K_{total}}q_k^{\natural} \sqrt{g_k^{\natural}} a_{k,i}\bm{h}_k+ \bm{z}_i, \quad i=1,2,\cdots,D,
\end{equation}
where
$\bm{h}_k\in \mathbb{C}^{M}$ is the channel vector of small-scale fading coefficients,
$g_k^{\natural}\in \mathbb{R}_{+}$ is the large-scale fading coefficient~(LSFC) and
$\bm{z}_i\in \mathbb{C}^{M}$ is the additive white Gaussian noise~(AWGN).
The user's activity is represented by $q_k^{\natural}\in\mathbb{B}$, i.e., $q_k^{\natural}=0$ for inactive and $q_k^{\natural}=1$ for active, and $\sum_{k=1}^{N}q_k^{\natural}=K$.
The problem we study is to recover the transmitted signal, which can be determined by the indices of positive $q_k^{\natural}$, at the BS with unknown $g_k^{\natural}$ and $\bm{h}_k$.

\subsection{Prior Work}

Various schemes have been proposed for URA in AWGN channels, such as \textit{T}-fold ALOHA~\cite{ordentlich2017low}, coded compressive sensing~\cite{amalladinne2020coded,amalladinne2022unsourced} and sparse regression code (SPARC)~\cite{fengler2019sparcs}. Considering the channel fading in a wireless system, the Reed-Muller based scheme~\cite{wang2022unsourced} were proposed for URA with a single antenna at BS. However, it does not apply to the multiple-input multiple-output (MIMO) case, where BS is equipped with multiple antennas. While the \textit{T}-fold ALOHA scheme~\cite{kowshik2019energy,kowshik2019short} and the differential coding scheme~\cite{mohammadkarimi2022massive} can be applied to the MIMO case, both of these schemes utilize clustering and successive interference cancellation (SIC) for data decoding, which incurs high computational complexity and delay when active user number $K$ is large. A pilot-based coherent scheme was proposed in~\cite{carvalho2017random,fengler2022pilot} for URA in the MIMO scenario, where orthogonal or non-orthogonal pilots are used for active detection and channel estimation. Then, maximum ratio combining using the estimated channels is employed for MIMO detection. It requires that the estimated channels vary slowly from the beginning of pilot transmission to the end of data transmission. That is, this pilot-based scheme requires coherence block of large channel uses (around $10^3$), which is far from practical in URA.

Alternatively, a concatenated coding scheme was proposed in~\cite{fengler2019sparcs,fengler2021non} for URA to use SPARC and tree code as inner and outer codes, respectively.
The inner SPARC use a common codebook matrix of size $D\times N$ for all the potential users, where up to $\log_2{N}$ information bits can be transmitted and $K\ll N$. \footnote{The inner SPARC codebook can support arbitrary number of potential users such that the exact $K_{total}$ is generally irrelevant. The parameter $N$ is directly related to the number of active user $K$ to ensure a certain level of sparsity for SPARC decoding.} To reduce the storage and computational complexity for SPARC involving large $N$, the tree code is used as outer code to enable dividing the original information bits into several sections that are separately encoded by SPARC with much smaller $N$. Other choices for the outer code include cyclic redundancy check and Reed-Solomon code~\cite{andreev2021reed}. At the receiver, non-coherent detections are used for decoding SPARC section by section and the outer tree decoder stitches different sections of messages together to produce original information messages.

For the inner SPARC, a series of recent works~\cite{liu2018massive,liu2018massive2} have formulated~\eqref{problem:primal} as a multiple measurement vector~(MMV) problem. The goal is to recover $q_k^{\natural}$ and $\bm{h}_k$ simultaneously while treating $g_k^{\natural}$ as either deterministic known quantities or as random quantities whose prior distribution is known. AMP type algorithms are applied as a popular approach for such a MMV problem~\cite{vila2011expect,liu2018massive}. However, the compressive-sensing-based approaches such as AMP~\cite{liu2018massive} and Orthogonal AMP~(OAMP)~\cite{ma2017orthogonal} performs poorly when the number of active users $K>D$~\cite{fengler2019sparcs,fengler2021non}. The above limitation means that the dimensions of the codebook matrix $D$ must scale linearly with $K$, which could be infeasible in practice.
Alternatively, covariance-based non-coherent estimators are proposed for the SPARC decoding, such as (relaxed) maximum-likelihood~(ML) and non-negative least squares~(NNLS), which are shown to have better empirical performance than AMP at the cost of higher computational complexity~\cite{haghighatshoar2018improved,fengler2021non}.\footnote{For simplicity, the (relaxed) maximum-likelihood will be referred as ``ML'' in the remaining of the work.}

\subsection{Contributions}
The paper is focused on the SPARC decoding with practical concerns. Essentially, existing AMP~\cite{liu2018massive}, OAMP~\cite{ma2017orthogonal} and covariance-based non-coherent estimators~\cite{haghighatshoar2018improved} require the channel to stay static over $D$ channel uses during the transmission of the codeword. However, due to physical environment constraints, the coherence block of the uplink channel is usually limited around hundreds of channel uses in most exciting wireless cellular systems~\cite{kowshik2019short,tse2005fundamentals}, where both AMP and OAMP deteriorates significantly~\cite{fengler2019sparcs}. Furthermore, the number of active users supported by compressive-sensing-based approaches, e.g., AMP~\cite{liu2018massive}, OAMP~\cite{ma2017orthogonal}, ISTA~\cite{daubechies2004iterative}, FISTA~\cite{beck2009fast}, ADMM~\cite{boyd2011distributed}, is bounded by $\widetilde{\mathcal{O}}(D)$~\cite{haghighatshoar2018improved,kueng2018robust}.\footnote{By $\widetilde{\mathcal{O}}(\cdot)$, we omit some logarithm factor in the remaining of the work.} On the other hand, the covariance-based approaches perform remarkably well for small $D$ and the number of supported users can achieve $\widetilde{\mathcal{O}}(D^2)$~\cite{haghighatshoar2018improved}. The main issue for covariance-based approaches are their high computational complexity.

In this paper, we concentrate on the covariance-based approaches and propose a two-stage method called accelerated maximum-likelihood which combines a one-step thresholding-based algorithm for dimension reduction and ML for refinement. Due to the dimension reduced in the first stage, the complexity of ML is significantly reduced while the overall performance still being maintained. Adequate numerical results are provided to support such claim including the link-level simulations over the existing cellular systems. On the other hand, theoretical analyses are provided for both the covariance-based sparse recovery model and the thresholding algorithm. For the former, the optimal upper bound for the number of active user $K$ is provided when convex relaxation is considered. For the latter, the one-step thresholding-based algorithm is theoretically guaranteed to contain all the active users with high probability when the number of BS antennas $M$ and channel uses $D$ are reasonably large.

%

The structure of this paper is as follows: The standard concatenated coding scheme of SPARC and tree code~\cite{fengler2021non} is introduced in~\cref{section:preliminaries} while our proposed algorithm is introduced in~\cref{section:algorithm}. Theoretical analyses for both the covariance-based sparse recovery model and the thresholding algorithm are provided in~\cref{section:theory}.
In~\cref{section:experiments}, simulations are provided to demonstrate the performance of the proposed algorithm. Finally in~\cref{section:conclusion}, this work is concluded.

\subsection{Notations}
For any $a\in \mathbb{C}$, we denote $\overline{a}$ the conjugate of $a$. For any vector $\bm{x}\in \mathbb{R}^{N}$ and any matrix $\bm{A}\in \mathbb{R}^{D \times N}$, $\bm{x}^{T}$ and $\bm{A}^{T}$ are their transpose, respectively. For $\bm{x}\in \mathbb{C}^{N}$ and $\bm{A}\in \mathbb{C}^{D \times N}$, $\bm{x}^{*}$ and $\bm{A}^{*}$ are their conjugate (Hermitian) transpose respectively. $\bm{A}(i,:),\bm{A}(:,j)$ are the $i$-th row and $j$-th column
of $\bm{A}$ respectively. For $a\in \mathbb{C}$, $(a)_+=\max(a,0)$. For a positive integer $n$, the notation $[n]$ represents $[n]=\left\{1,2,\cdots,n\right\}$. For $\bm{x}\in \mathbb{R}^{N}$, $x_{\min}$, $x_{\max}$ are the smallest and largest components in magnitude of $\bm{x}$, respectively. For $\bm{\gamma}\in \mathbb{R}^{N}_+$, it means that $\bm{\gamma}\in \mathbb{R}^{N}$ and each component of $\bm{\gamma}$ is non-negative, i.e., $\gamma_i\ge 0$, $\forall i\in[N]$.
For any $\bm{\gamma}\in\mathbb{R}^N$, the notation $\mathrm{diag}(\bm{\gamma})$ means the $N\times N$ diagonal matrix whose diagonal entries are given by $\{\gamma_i\}_{i=1}^N$.
For a set $\mathcal{S}$, $|\mathcal{S}|$ denotes the number of elements in $\mathcal{S}$. For any positive integer $s$ and any vector $\bm{x}\in \mathbb{C}^{N}$, $\mathcal{H}_{s}$ denotes the hard thresholding operator that keeps only the largest $s$ components (in magnitude) of $\bm{x}$ and set the others to be zero. For matrix $\bm{A}\in\mathbb{C}^{D\times N}$, $\mathrm{vec}(\bm{A})$ denotes the $DN\times 1$ vector obtained by stacking the columns of $\bm{A}$.

\section{Preliminaries}
\label{section:preliminaries}

\begin{figure*}[tbh]
       \includegraphics[width=\textwidth]{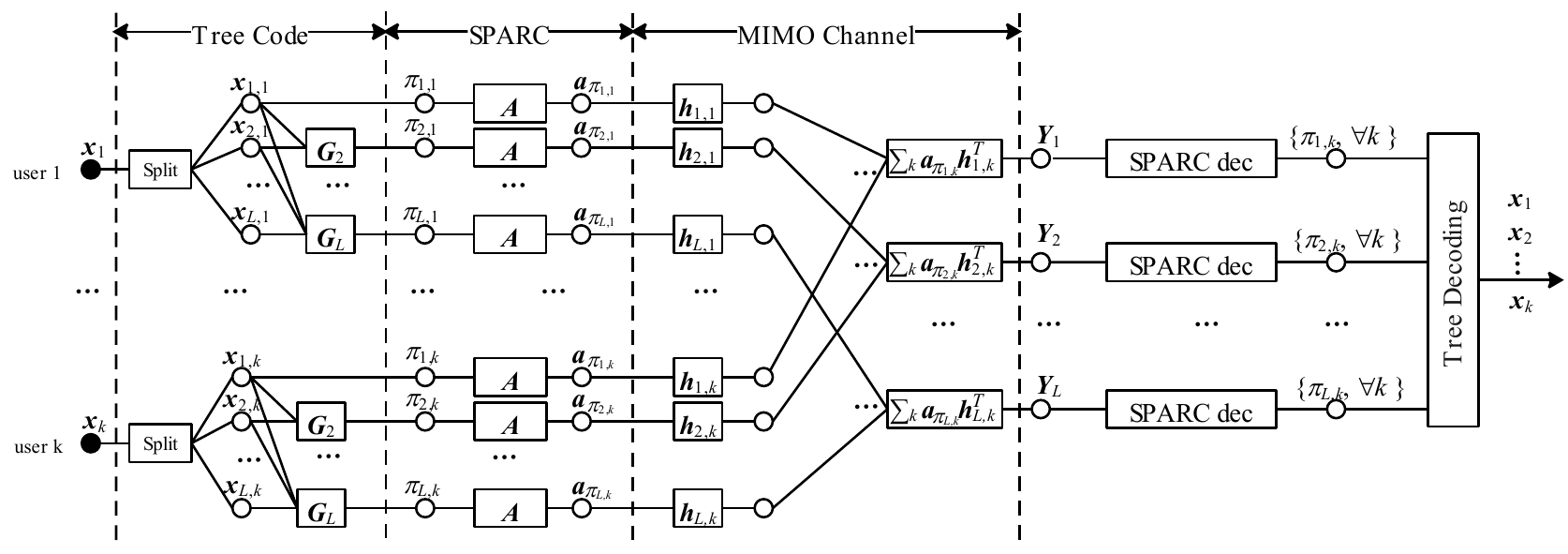}
       \centering
       \caption{\label{fig:URA_system}The concatenated coding scheme of SPARC and tree code for URA.}
\end{figure*}

In the standard concatenated coding scheme for the URA~\cite{fengler2019sparcs,fengler2021non} in~\cref{fig:URA_system}, a frame of $D_{total}$ channel uses is shared by arbitrary number of potential users and only $K\ll N$ users are active. Each active user transmits an information message of $B$ bits, i.e., $\bm{x}_k\in\mathbb{B}^{B}$.

\subsection{Encoding}
\label{subsec:std-encoding}

The information message is first encoded by an outer tree code~\cite{amalladinne2020coded} as shown in~\cref{fig:URA_system}. Specifically, the $k$-th user's information message $\bm{x}_k$ is split as $\bm{x}_{1,k}, \bm{x}_{2,k},\cdots,\bm{x}_{L,k}$. The $L$ sections are of length $b_1,b_2,\cdots,b_L$, where $\sum_{l=1}^{L}b_l=B$ and $b_1=J$ and $b_l<J$ for $l=2,\cdots,L$. Each section is augmented to size $J$ by appending $p_l=J-b_l$ parity bits. The coding rate of tree code is $R_{tree}=\frac{B}{L\cdot J}$. For section $l=2,3,\cdots,L$, $p_l$ parity bits are generated out of the previous $l-1$ sections of information bits using a random binary matrix $\bm{G}_{l}\in\mathbb{B}^{p_l\times m_l}$, where $m_l=\sum_{\ell=1}^{l-1}b_{\ell}$. The random binary matrices $\{\bm{G}_{l}\}_{l=2}^L$ are identical for all users. A total number of $P=\sum_{l=2}^{L}p_l$ parity bits are used to link $L$ sections of information bits together to produce the original information message at the receiver. The $L$ coded messages of user $k$ is denoted as $\pi_{1,k},\pi_{2,k},\cdots,\pi_{L,k}$, where $\pi_{l,k}\in[2^J]$ is a decimal representation of a $J$-bit binary message.

For the inner code, each user will apply SPARC individually to encode each $\pi_{l,k}$ into a length $D=D_{total}/L$ codeword separately~\cite{fengler2021non}. Let $\bm{A}=[\bm{a}_{1},\bm{a}_{2},\cdots,\bm{a}_{N}]\in\mathbb{C}^{D\times N}$ with $N=2^J$ be the codebook matrix that is shared by all potential users. It is assumed that the codebook matrix $\bm{A}$ is a column-normalized random matrix, e.g., the matrix whose columns are i.i.d. random variables drawn uniformly from the sphere, Bernoulli random matrix, sub-sampled Fourier matrix, etc. The sub-sampled Fourier matrix is preferred from the computational point of view. It has been empirically shown to have no performance loss in terms of reconstruction error when compared with the random i.i.d Gaussian ones in the context of compressive sensing~\cite{barbier2017approximate}, which is further confirmed by the replica analysis~\cite{wen2014analysis,hou2022sparse}. Each user $k$ chooses one column of the codebook based on its own $J$-bit coded message $\pi_{l,k}$ as its codeword $\bm{a}_{\pi_{l,k}}\in \mathbb{C}^{D}$. That is, the index $\pi_{l,k}$ of the selected column of $\bm{A}$ represents the $J$-bit message. The coding rate of SPARC is $R_{SPARC}=\frac{J}{D}$. Concatenated with tree code, the information rate of each user is given by $R=R_{SPARC}\times R_{tree}=\frac{B}{D_{total}}$ and the sum rate of $K$ active users is given by  $R_{total}=\frac{B\cdot K}{D_{total}}$.

\subsection{Decoding}
\label{subsec:std-decoding}

For a BS with $M$ antennas and by following the system model in~\eqref{problem:primal}, let the received signals matrix $\bm{Y}=\left[\bm{y}_1,\bm{y}_2,\cdots,\bm{y}_{D}\right]^T\in \mathbb{C}^{D \times M}$, the channel coefficients matrix $\bm{H}=\left[\bm{h}_1,\bm{h}_2,\cdots,\bm{h}_{N}\right]^T\in \mathbb{C}^{N \times M}$ and the AWGN matrix  $\bm{Z}=\left[\bm{z}_1,\bm{z}_2,\cdots,\bm{z}_{D}\right]^T\in \mathbb{C}^{D \times M}$. Therefore, the received signal for each section can be written in matrix form as\footnote{Since the inner encoder and decoder are identical for each section, the subscription ``$l$'' is omitted in~\eqref{problem:matrix} for simplicity.}
\begin{equation}\label{problem:matrix}
       \bm{Y}=\bm{A}\sqrt{\bm{\Gamma}^{\natural}}\bm{H}+\bm{Z},
\end{equation}
where the entries of $\bm{H}$ and $\bm{Z}$ are mutually independent and distributed as $\mathcal{CN}(0,1)$ and $\mathcal{CN}(0,\sigma^2)$, respectively, and $\bm{\Gamma}^{\natural}\in \mathbb{C}^{N\times N}$ is a diagonal matrix whose diagonal entries are given by $\{q_k^{\natural} g_k^{\natural}\}_{k=1}^N$ in~\eqref{problem:primal}. We consider the case that at each time slot, the number of active users $K$ is much smaller than the codebook size, i.e. $K\ll N$. Let $\gamma_k^{\natural}:=q_k^{\natural} g_k^{\natural}$ be the diagonal elements of $\bm{\Gamma}^{\natural}$, then we have
\begin{equation*}
       \lV \bm{\gamma}^{\natural}\rV_0\le K,\; K\ll N.
\end{equation*}
Furthermore, the non-zero entries of $\bm{\gamma}^{\natural}$ are determined by coded messages $\pi_{l,k}$, i.e., its $\pi_{l,k}$-th entry is $q_{\pi_{l,k}}^{\natural} g_{\pi_{l,k}}^{\natural}$ and $0$ otherwise. The set of the coded messages of $K$ users in each section is denoted by $\mathrm{supp}(\bm{\gamma}^{\natural})$, which is the set of the index of the non-zero entries of $\bm{\gamma}^{\natural}$.

The ML estimator for~\eqref{problem:matrix} is to obtain the estimate $\bm{\Gamma}^{\natural}$ (or $\bm{\gamma}^{\natural}$) from the the sample covariance
\begin{equation}\label{Y:samplesigma}
       {\bm{\widehat{\Sigma}_y}}=\frac{1}{M}\bm{Y}\bm{Y}^*.
\end{equation}
The ML algorithm will output an estimate of $\bm{\gamma}^{\natural}$ as $\bm{\hat{\gamma}}$, where the $j$-th entry $\hat{\gamma}_{j}$ denotes the possibility of $\bm{a}_{j}$ being transmitted by an active user and $\hat{\gamma}_{j}=0$ if this codeword is not selected by any active user.

While active user number $K$ is unknown at the BS, it is reasonable to assume that the maximal number of active user $T$ is known~\cite{kowshik2019short,ordentlich2017low}. The ML estimator generates $\mathcal{L}(\bm{Y}_l)=\{\hat{\pi}_{l,1},\hat{\pi}_{l,2},\cdots,\hat{\pi}_{l,T}\}$ for each section, which consists the indices of the largest $T$ components in $\bm{\hat{\gamma}}$. There are a total of $L$ sets of messages, i.e., $\{\mathcal{L}(\bm{Y}_l)\}_{l=1}^L$, each of which contains $T$ messages and each message conveys $J$ coded bits. For each message, the tree decoding aims to link its $L$ sections together with parity bits and remove falsely detected coded messages simultaneously.
In this paper, our focus is concentrated on the inner decoder while details for tree decoding can be found in~\cite[Section IV-C]{amalladinne2020coded}.

\section{The Proposed Two-Stage Algorithm}\label{section:algorithm}
In this section, a covariance-based approach is proposed to decode SPARC in two stages: a one-step thresholding-based algorithm for rough support recovery and a ML estimator for refinement.
First, \cref{problem:matrix} is reformulated into a sparse vector recovery problem, where the thresholding-based algorithms are introduced. Then, we give a brief review about the ML estimator which numerically performs well in the SPARC decoding. Last, the two-stage method is introduced to reduce the overall computational cost while maintaining similar performance as ML.

\subsection{Recovery of Sparse Vectors from Covariance Matrix}
For i.i.d. Gaussian channel vectors $\bm{h}_k\in \mathbb{C}^{M}$ and additive Gaussian noise $\bm{z}_i\in \mathbb{C}^{M}$, the columns of $\bm{Y}$ are independent samples from a multivariate complex Gaussian distribution as $\bm{Y}(:,j)\sim \mathcal{CN}(0,\bm{\Sigma_y})$. According to~\eqref{problem:matrix}, the covariance matrix can be calculated as~\cite{chen2019cov}
\begin{equation}\label{Y:truesigma}
       \bm{\Sigma_y}=\bm{A}\bm{\Gamma}^{\natural}\bm{A}^*+\sigma^2\bm{I}_D.
\end{equation}
Recall ${\bm{\widehat{\Sigma}_y}}=\frac{1}{M}\bm{Y}\bm{Y}^*$ is the sample covariance.
Note that
${\bm{\widehat{\Sigma}_y}}\to \bm{\Sigma_y}$ if $ M\to \infty$. Thus for sufficiently large $M$, $\bm{\gamma}^{\natural}$ can be approximately estimated from the sample covariance matrix ${\bm{\widehat{\Sigma}_y}}$ by solving the follow problem: find $\bm{\gamma}^{\natural}\in\mathbb{R}^N_+$ that satisfies
\begin{equation}\label{problem:p0}
       \bm{\widehat{\Sigma}_y}= \bm{A}\mathrm{diag}(\bm{\gamma}^{\natural})\bm{A}^*+\sigma^2\bm{I}_D+\bm{\Delta},\; \text{s.t.} \; \lV\bm{\gamma}^{\natural}\rV_0\le K,
\end{equation}
where $\bm{\Delta}=\bm{\widehat{\Sigma}_y}-\bm{\Sigma_y}$ is the unknown error matrix between the sample and true covariance matrix. To bound $\bm{\Delta}$, it has been shown in \cite[Theorem 11]{fengler2021non} that
\begin{equation}\label{bound:Pdelta}
       \mathbb{P}\left(\lV \bm{\Delta}\rV_F\le \frac{c_{\varepsilon}D\left(\lV\bm{\gamma}^{\natural}\rV_1+\sigma^2\right)}{\sqrt{M}}\right)\ge 1-\varepsilon
\end{equation}
for $\lV\bm{a}_k\rV^2_2=D$ and sufficiently large $M$ (see also Lemma~\ref{bound:Delta}).
Here, $c_{\varepsilon}=\sqrt{c\log \left(\frac{(eD)^2}{\varepsilon}\right)}$ , where $\varepsilon\in (0,1)$ is arbitrary and $c$ is an universal constant. It then implies that $\lV\bm{\Delta}\rV_F$ is small when $M$ is sufficiently large. In view of this, the problem~\eqref{problem:p0} can be regarded as a robust sparse vector recovery problem. Hence, it is natural to consider solving the following $\ell_0$-minimization problem
\begin{equation}\label{problem:l0recovery}
       \mathop{\mathrm{minimize}}\limits_{\bm{\gamma}\in\mathbb{R}^N_+}~\lV\bm{\gamma}\rV_0,~\text{s.t.}~ \lV \bm{\Sigma}(\bm{\gamma})-{\bm{\widehat{\Sigma}_y}}\rV_F\le \eta,
\end{equation}
where
\begin{equation*}
       \bm{\Sigma}(\bm{\gamma}):=\bm{A}\mathrm{diag}(\bm{\gamma})\bm{A}^*+\sigma^2\bm{I}_D,
\end{equation*}
and
$\eta:=\frac{c_{\varepsilon}D\left(\lV\bm{\gamma}^{\natural}\rV_1+\sigma^2\right)}{\sqrt{M}}$ according to~\eqref{bound:Pdelta}. If the linear equation in \eqref{problem:p0} is underdetermined, it is widely known that the $\ell_0$-minimization problem \eqref{problem:l0recovery} is NP-hard in general  (see \cite[Theorem 2.17]{foucart2013}). Thus, one may consider its convex relaxation in form of $\ell_1$-minimization as
\begin{equation}\label{problem:p1}
       \mathop{\mathrm{minimize}}\limits_{\bm{\gamma}\in\mathbb{R}^N}~\lV\bm{\gamma}\rV_1,~\text{s.t.}~ \lV \bm{\Sigma}(\bm{\gamma})-{\bm{\widehat{\Sigma}_y}}\rV_F\le \eta,
\end{equation}
where the non-negativity in $\bm{\gamma}$ is also relaxed. The recovery guarantees of the basis pursuit approach~\eqref{problem:p1} can be found in Section~\ref{section:theoryforSparseRe}, showing that \eqref{problem:p1} has exact recovery of $\bm{\gamma}^{\natural}$ with near-optimal scaling law in $K$ and $D$.
Problem~\eqref{problem:p1} is convex and many algorithms can be applied, e.g., ISTA~\cite{daubechies2004iterative}, FISTA~\cite{beck2009fast}, ADMM~\cite{boyd2011distributed}. Though convergence of such algorithms are well studied,  there is few results that~\eqref{problem:p1} can be exactly solved in a (very few) fixed number of iterations. In this work, we consider the non-convex IHT algorithm in the following sub-sections. Theoretically, we show that the IHT algorithm can be guaranteed to have exact support recovery in just one step -- meaning that the algorithm can be designed to be non-iterative and super efficient.

\subsection{One-Step IHT Algorithm}
In this sub-section, we apply a first-order gradient-type non-convex algorithm for~\eqref{problem:p0}. To begin with, we define a matrix $\bm{\A}\in \mathbb{C}^{D^2\times N}$ where its $k$-th column is given by
\begin{equation}\label{def:calA}
       \bm{\A}(:,k):=\mathrm{vec}(\bm{a}_k\bm{a}_k^*),~k\in [N].
\end{equation}
Relaxing the non-negativity in $\bm{\gamma}^{\natural}$, the problem~\eqref{problem:p0} can be reformulated as
\begin{equation}\label{problem:CS}
       \bm{u}=\bm{\A}\bm{\gamma}^{\natural}+\bm{\delta},~\text{s.t.}~ \lV\bm{\gamma}^{\natural}\rV_0\le K,
\end{equation}
where $\bm{u}:=\mathrm{vec}(\bm{\widehat{\Sigma}_y}-\sigma^2\bm{I}_{D})$ and $\bm{\delta}:=\mathrm{vec}(\bm{\Delta})$.
Consider the standard least squares fitting for~\eqref{problem:CS}. We then have the following minimization problem
\begin{equation}\label{eq:conslsint}
       \mathop{\mathrm{minimize}}\limits_{\lV\bm{\gamma}\rV_0\le K}~f_{ls}(\bm{\gamma}),~\mbox{where}~ f_{ls}(\bm{\gamma}):=\frac{1}{2D^2}\lV\bm{u}-\bm{\A}\bm{\gamma}\rV_{2}^{2}.
\end{equation}
Therefore, we can directly apply the projected gradient descent (PGD) algorithm for \eqref{eq:conslsint}. When the hard thresholding operator is employed for projection, the PGD is equivalent to the IHT algorithm~\cite{blumensath2009iterative}. The update at $t$-th iteration of IHT is given by
\begin{align}\label{iter:IHT}
       \bm{\gamma}_{t+1}
       =\mathcal{H}_{\hat{K}}\left(\bm{\gamma}_{t}-\alpha \nabla f_{ls}(\bm{\gamma}_{t})\right)
       =\mathcal{H}_{\hat{K}}\left(\bm{\gamma}_{t}+\frac{\alpha}{D^2} \mathcal{\bm{\A}}^*\left(\bm{u}-\bm{\A}\bm{\gamma}_{t}\right)\right),\;t\ge 0,
\end{align}
where $\alpha>0$ is the step size and the hard thresholding operator $\mathcal{H}_{\hat{K}}$ is to project the update into the space $\{\bm{\gamma}\in\mathbb{R}^N:\lV\bm{\gamma}\rV_0\le \hat{K}\}$. For different codebooks $\bm{A}$, the linear convergence guarantees of IHT in~\eqref{iter:IHT} are provided in Section~\ref{section:theoryforSparseRe}.

Notice the algorithm rely on a proper choice of step size and iteration number (see \cref{theorem:IHTconv}).  In fact, we are aiming at support recovery in this work, thus the recovery of  $\bm{\gamma}^\natural$ is not essentially required (see \cref{fig:URA_system}). Therefore, we consider a much more simple and efficient non-iterative method -- the one-step IHT algorithm, see Algorithm~\ref{alg:oneIHT}.

\begin{algorithm}[tbh]
       \DontPrintSemicolon
       \SetKwInOut{Initialize}{Initialize}
       \SetAlCapHSkip{0em}
       \caption{One-Step IHT Algorithm}\label{alg:oneIHT}
       \Indmm\Indmm
       \Initialize { $\bm{\gamma}_0=\bm{0}$.}
       \Indpp\Indpp
       $\bm{\gamma}_1=\mathcal{H}_{\hat{K}}\left(\bm{\gamma}_{0}+\frac{1}{D^2} \mathcal{\bm{\A}}^*\left(\bm{u}-\bm{\A}\bm{\gamma}_{0}\right)\right)$\;
       $S_1 = \mathrm{supp}(\bm{\gamma}_1)$\;
       \Indmm\Indmm
       \KwOut{The resulting estimate  $S_1$.}
       \Indpp\Indpp
\end{algorithm}

Furthermore, for one step of IHT, we show that w.h.p. it is guaranteed to have exact support recovery (i.e. $S_1 = \mathrm{supp}(\bm{\gamma}^\natural)$) in \cref{section:theoryforalg2}. To further improve the performance of the algorithm, we combine the one-step IHT with a ML estimator in the following sections.


\subsection{The ML Estimation via Coordinate-Wise Optimization}

In this subsection, we give a brief review about the ML estimation for the SPARC decoding and the corresponding iterative solver introduced in \cite{haghighatshoar2018improved,fengler2021non}. From \eqref{Y:truesigma}, we consider the following negative log-likelihood cost function
\begin{align}\label{def:functionf}
       f\left(\bm{\gamma}\right) & :=-\frac{1}{M}\log p(\bm{Y}|\bm{\gamma})= -\frac{1}{M}\sum_{j=1}^M\log p(\bm{Y}(:,j)|\bm{\gamma})
       \propto\log|\bm{\Sigma}(\bm{\gamma})|+\mathrm{tr}(\bm{\Sigma}(\bm{\gamma})^{-1}\bm{\Sigma_y}).
\end{align}
By relaxing the sparse constrained problem into a non-negative constrained problem, the relaxed ML estimator is given by
\begin{equation}\label{def:nonNeML}
       \bm{\gamma}^*_r=\mathop{\arg\, \min}\limits_{\bm{\gamma}\in \mathbb{R}^N_+}f\left(\bm{\gamma}\right).
\end{equation}
A gradient-type algorithm based on the coordinate-wise optimization for solving~\eqref{def:nonNeML} can be summarized as follows. For each coordinate $k\in [N]$, define the scalar function
\begin{equation*}
       f_k\left(d\right)= \log|\bm{\Sigma}+d\bm{a}_k\bm{a}_k^*|
       +\mathrm{tr}\left[(\bm{\Sigma}+d\bm{a}_k\bm{a}_k^*)^{-1}\bm{\Sigma_y}\right],
\end{equation*}
where $\bm{\Sigma}:=\bm{\Sigma}(\bm{\gamma})$ and $\gamma\in\mathbb{R}^N$. By setting the derivative of $f_k\left(d\right)$ as zero, i.e., $f_{k}'(d)=0$, the solution is given by
\begin{equation*}
       \hat{d}_k=\frac{\bm{a}_k^*\Sigma^{-1}\bm{\widehat{\Sigma}_y}\bm{\Sigma}^{-1}\bm{a}_k-\bm{a}_k^*\bm{\Sigma}^{-1}\bm{a}_k}{\left(\bm{a}_k^*\bm{\Sigma}^{-1}\bm{a}_k\right)^2}.
\end{equation*}
Then one step of the coordinate-wise projected gradient descent update is given by: for all $k\in [N]$,
\begin{subequations}\label{alg:ML}
       \begin{align}
              \gamma_k                 & \leftarrow\left(\gamma_k+d_k\right)_+,~\text{where}~d_k=\hat{d}_k, \label{alg:ML_scale} \\
              \bm{\Sigma}(\bm{\gamma}) & \leftarrow\bm{\Sigma}(\bm{\gamma})+d_k\bm{a}_k\bm{a}_k^*, \label{alg:ML_vector}
       \end{align}
\end{subequations}
where $\left(\cdot\right)_+$ is to project the update onto the non-negative space.  Further, to avoid the expensive computational cost in calculating $\Sigma^{-1}$, the rank-one update in~\eqref{alg:ML_vector} can be replaced by
\begin{equation*}
       \bm{\Sigma}^{-1}\leftarrow \bm{\Sigma}^{-1}-\frac{d_k \bm{\Sigma}^{-1}\bm{a}_k \bm{a}_k^*\bm{\Sigma}^{-1}}{1+d_k \bm{a}_k^*\bm{\Sigma}^{-1}\bm{a}_k},
\end{equation*}
where the Sherman-Morrison update is applied. The iterative algorithm for solving \eqref{def:nonNeML} is then obtained by repeating the above coordinate-wise update.

Alternatively, if the standard least squares fitting for \eqref{problem:p0} is considered and the sparse constraint is relaxed to the non-negative constraint, the NNLS is then to solve the following convex problem~\cite{fengler2021non,haghighatshoar2018improved,slawski2013non}
\begin{equation}\label{def:NNLS}
       \bm{\gamma}^*_N=\mathop{\arg\, \min}\limits_{\bm{\gamma}\in \mathbb{R}^N_+}\lV \bm{\Sigma}(\bm{\gamma})-{\bm{\widehat{\Sigma}_y}}\rV_F.
\end{equation}
The NNLS estimator is guaranteed to identify the activity of $\widetilde{\mathcal{O}}(D^2)$ active users~\cite{fengler2021non,haghighatshoar2018improved}. Comparing with NNLS estimator in~\eqref{def:NNLS}, the ML estimator in~\eqref{def:nonNeML} is non-convex and thus the corresponding gradient algorithm falls short in theoretical guarantees. However,  ML estimator empirically outperforms the former significantly. In experiments, ML can successfully identify the activity of $\widetilde{\mathcal{O}}(D^2)$ active users. Therefore, the coordinate-wise iterative algorithm for ML in~\eqref{def:nonNeML} is chosen for the refinement stage in our algorithm.

\subsection{The Proposed Algorithm: A Two-Stage Method}
In this section, we propose a two-stage method based on the thresholding-based algorithm and ML estimator. For ML estimator, the existing gradient-type coordinate-wise minimization approaches are iterative algorithms on the entire index set $\{1,2,\cdots,N\}$. By utilizing the sparsity in $\bm{\gamma}^{\natural}$, we introduce a fast coordinate-wise minimization method which reduces the computational cost dramatically in the first phase of the iteration. That is, we estimate a set $S^{0}$ that is most likely to contain the support of $\bm{\gamma}^{\natural}$ in the first stage of the algorithm. In the second stage, we adopt the coordinate minimization to the following problem
\begin{equation}\label{alg:ML_opt}
       \hat{\bm{\gamma}} = \mathop{\arg\, \min}\limits_{
              \substack{\bm{\gamma}\in  \mathbb{R}^{N}_{+},\\
                     \mathrm{supp}(\bm{\gamma})\subset S^{0}}
       }~f\left(\bm{\gamma}\right),
\end{equation}
where $f\left(\bm{\gamma}\right)$ is defined in \eqref{def:functionf}. To find out the indices that are most likely to be the support of $\bm{\gamma}^\natural$ in the first stage,  we define
\begin{equation}\label{def:Yk}
       \Y_k :=\sum_{1\le i,j \le D}\overline{a}_{k,i}a_{k,j} \widehat{\Sigma}_y(i,j),~\forall k\in [N],
\end{equation}
where $\widehat{\Sigma}_y(i,j)$ is the $(i,j)$-th component of $\bm{\widehat{\Sigma}_y}$. Recall that the channel and AWGN vectors are i.i.d. Gaussian. Thus we have the following proposition.

\begin{proposition}\label{prop:expZk}
       Assume that $\{\bm{a}_k\}_{k=1}^{N}$ are i.i.d. random variables drawn uniformly from the sphere of radius $\sqrt{D}$.  For $\Y_k$ defined in \eqref{def:Yk}, we have
       \begin{align}\label{def:EYk}
              E[\Y_k]=
              \left\{
              \begin{aligned}
                      & \psi(D,\bm{\gamma}^{\natural},\sigma),~                                      & \forall k\notin  S^{\natural}, \\
                      & \psi(D,\bm{\gamma}^{\natural},\sigma)+D^2\gamma^{\natural}_{k},~ & \forall k\in  S^{\natural},    \\
              \end{aligned}
              \right. \nonumber
       \end{align}
       where  $S^{\natural}=\mathrm{supp}(\bm{\gamma}^{\natural})$ and
       $\psi(D,\bm{\gamma}^{\natural},\sigma):=\sum_{\ell\in S^{\natural} }\frac{D^2}{D+1}\gamma^{\natural}_{\ell}+D\sigma^2.$\\
       Proof: See Appendix~\ref{proof:propexpZk}.
\end{proposition}
The Proposition~\ref{prop:expZk} indicates that statistically the indices of the larger components of $\{\Y_k\}_{k=1}^N$ are the support of $\bm{\gamma}^\natural$. Thus, it is natural to select larger components of $\{\Y_k\}_{k=1}^N$ as an initial guess of $S^{\natural}$, and we have the following Remark~\ref{remark:eqYandIHT}.
\begin{remark}\label{remark:eqYandIHT}
       In fact, selecting larger components of $\{\Y_k\}_{k=1}^N$ is equivalent to the one-step IHT in Algorithm~\ref{alg:oneIHT}: Notice that the estimation of one-step IHT is given by $\bm{\gamma}_1\propto \H_{\hat{K}}\left(\bm{\A}^*\bm{u}\right)$. For $\lV\bm{a}_k\rV_2^2=D$, by the definition of $\bm{u}$ we have
       \begin{align*}
              \Y_k =\sum_{1\le i,j \le D}\overline{a}_{k,i}a_{k,j}\left(\widehat{\Sigma}_y(i,j)-\sigma^2\delta_{i,j}\right)+\sigma^2D
              =\bm{\A}_k^*\bm{u}+\sigma^2D,
       \end{align*}
       where $\bm{\A}_k$ is the $k$-th column of $\bm{\A}$. Therefore, denoting $\bm{\Y}:=[\Y_1,\Y_2,\cdots,\Y_N]$, we then have
       \begin{equation}\label{eq:equiIHT}
              \mathrm{supp}\left(\bm{\gamma}_1\right)= \mathrm{supp}\left(\H_{\hat{K}}\left(\bm{\Y}\right)\right).
       \end{equation}
\end{remark}
\noindent Further, instead of finding the largest $\hat{K}$ components of $\{\Y_k\}_{k=1}^N$, a threshold $\overline{\Y}=\frac{1}{N}\sum_{k=1}^N\Y_k$ could be used to select the support, which is due to that only a small number of  $\{\Y_k\}_{k=1}^N$ are statistically larger. Thus the choice of parameter $\hat{K}$ is no longer required. Next, the accelerated ML estimator in~\eqref{alg:ML_opt} is employed for refinement in the second stage. The proposed two-stage method is summarized in Algorithm~\ref{alg:two_stage}.
\begin{algorithm}[tbh]
       \DontPrintSemicolon
       \SetKwInOut{Initialize}{Initialize}

       \SetAlCapHSkip{0em}
       \caption{Accelerated Maximum-Likelihood for SPARC via Coordinate-Wise Optimization}\label{alg:two_stage}

       \Indmm\Indmm
       \KwIn{$\left\{\bm{a}_k \right\}_{k=1}^{N}$, sample covariance matrix $\bm{\widehat{\Sigma}_y}=\frac{1}{M}\bm{Y}\bm{Y}^*$.}
       \Initialize {$\bm{\Sigma}=\sigma^2\bm{I}_D$, $\bm{\gamma}=0$, $\rho$ (e.g., $\rho=1$).}
       \Indpp\Indpp
       \label{step:threshold}Compute $\{\Y_k\}_{k=1}^N$ in~\eqref{def:Yk} and the set $S^{0}=\{k: \Y_k\ > \rho \overline{\Y}\}$, where $\overline{\Y}=\frac{1}{N}\sum_{k=1}^N\Y_k$.\;
       Set $\gamma_{\ell}=0$ for $\ell \notin S^{0}$.\;
       \For{$t=1,2,\cdots$}{\label{step:ML_st}
              \For(\tcc*[f]{Traverse all indices in the set $S^{0}$ randomly.}){$\ell\in S^{0}$}{
                     $d_\ell=\frac{\bm{a}_\ell^*\bm{\Sigma}^{-1}\bm{\widehat{\Sigma}_y}\bm{\Sigma}^{-1}\bm{a}_\ell-\bm{a}_\ell^*\bm{\Sigma}^{-1}\bm{a}_\ell}{\left(\bm{a}_\ell^*\bm{\Sigma}^{-1}\bm{a}_\ell\right)^2}$,\;
                     $\gamma_\ell \leftarrow \left(\gamma_\ell+d_\ell\right)_{+}$,\;
                     $\bm{\Sigma}^{-1} \leftarrow \bm{\Sigma}^{-1}-\frac{d_\ell \bm{\Sigma}^{-1}\bm{a}_\ell \bm{a}_\ell^*\bm{\Sigma}^{-1}}{1+d_\ell \bm{a}_\ell^*\bm{\Sigma}^{-1}\bm{a}_\ell}$.\;
              }
       }\label{step:ML_end}
       \Indmm\Indmm
       \KwOut{The resulting estimate $\bm{\gamma}$.}
       \Indpp\Indpp
\end{algorithm}

More specifically, step~\ref{step:threshold} in Algorithm~\ref{alg:two_stage} is the hard-thresholding stage. It computes a subset $S^{0}\subset [N]$ such that $S^{\natural}\subset S^{0}$ with high probability. In the calculation of $S^{0}=\{k: \Y_k\ > \rho \overline{\Y}\}$, the parameter $\rho$ decides the size of $S^{0}$ and is usually set to be around $1$. Steps~\ref{step:ML_st}-\ref{step:ML_end} in Algorithm~\ref{alg:two_stage} are the accelerated ML Stage, which adopts the coordinate-wise optimization for~\eqref{alg:ML_opt} as \eqref{alg:ML} in the subset $S^{0}$. In the next section, this algorithm will be analyzed thoroughly.

\section{Theoretical Results}
\label{section:theory}

\subsection{Theoretical Guarantees of Sparse Recovery}\label{section:theoryforSparseRe}
Let $\tilde{\bm{\gamma}}$ be the solution to the convex optimization problem in \eqref{problem:p1}, Theorem~\ref{theory:l1norm} below guarantees that the sparse vector recovery model in \eqref{problem:p1} admits a proper approximation to the underlying vector $\bm{\gamma}^\natural$ provided that $M$ is sufficiently large.
\begin{theorem}\label{theory:l1norm}
       Let $\varepsilon>0$ and assume that $\{\bm{a}_k\}_{k=1}^{N}$ are i.i.d. random variables drawn uniformly from the sphere of radius $\sqrt{D}$. For $\tilde{\bm{\gamma}}$ being the solution to \eqref{problem:p1}, we have
       \begin{align}
              \mathbb{P}\left[
                     \lV\tilde{\bm{\gamma}}-\bm{\gamma}^{\natural}\rV_1\le \frac{c_1\left(\lV\bm{\gamma}^{\natural}\rV_1+\sigma^2\right)}{\sqrt{M}}\sqrt{\log \left(\frac{(eD)^2}{\varepsilon}\right)} \right]
              \ge  1-\varepsilon-e^{-c_2D}
       \end{align}
       provided that
       \begin{equation*}
              K\le \frac{c_3 D(D-1)}{\log^2(eN/D)}\quad and \quad M\ge c\log \left(\frac{(eD)^2}{\varepsilon}\right),
       \end{equation*}
       where  $c_1,c_2,c_3,c$ are universal positive constants.
\end{theorem}
\begin{proof}
       See Appendix~\ref{proof:theoreml1rec}.
\end{proof}
A direct consequence of Theorem~\ref{theory:l1norm} is that
\begin{equation*}
       \lV\tilde{\bm{\gamma}}-\bm{\gamma}^{\natural}\rV_1\le c_1\varepsilon
\end{equation*}
holds as long as $K\le \mathcal{O}\left(D^2/\log^2(\frac{N}{D})\right)$ and
$ M\ge \left(\frac{K\gamma^\natural_{\max}+\sigma^2}{\varepsilon}\right)^2\log\left(\frac{(eD)^2}{\varepsilon}\right)$. Therefore, for $K\le \mathcal{O}\left(D^2/\log^2(\frac{N}{D})\right)$ and sufficiently large $M$, we will have arbitrary small $\ell_1$ error between the recovered vector and the ground truth. Hence, it is possible for the covariance-based model to successfully detect at most $\widetilde{\mathcal{O}}(D^2)$ active users. Further, the covariance-based sparse recovery model \eqref{problem:p1} imposes a loose requirement on channel uses $D$ comparing with the traditional compressive sensing approach, which aims to recover $\sqrt{\bm{\Gamma}^{\natural}}\bm{H}$ from~\eqref{problem:matrix}~\cite{vila2011expect,liu2018massive}. Given $K$ active users, Theorem~\ref{theory:l1norm} states that $D=\widetilde{\mathcal{O}}(\sqrt{K})$ is necessary. For comparison purposes, the traditional compressive sensing approach requires $D=\widetilde{\mathcal{O}}(K)$~\cite{kueng2018robust}.

When $M\to \infty$, we can see from~\eqref{bound:Pdelta} that $\lV\bm{\delta}\rV_2\to 0$, and \eqref{problem:p0} or \eqref{problem:CS} is reduced to find the $K$ sparse $\bm{\gamma}^{\natural}$ from $\bm{u}=\bm{\A}\bm{\gamma}^{\natural}$. Therefore, to ensure that the linear system admits a unique solution (i.e., $\bm{\gamma}^{\natural}$ be identifiable via the system), theoretically it requires each $2K$ columns of $\bm{\A}$ being linearly independent. Thus, the channel uses $D$ should satisfy $D^2\ge 2K$ to recover every $K$ sparse vector. It implies that the bound in $D$ and $K$ given in Theorem~\ref{theory:l1norm} is near optimal.

Now we investigate the theoretical performance of the thresholding-based algorithms for \eqref{problem:CS}. Due to its non-convexity, the convergence guarantee of IHT in~\eqref{iter:IHT} is not trivial. To ensure successful recovery, one common assumption is that the sensing matrix (e.g., $\bm{\A}$ in \eqref{problem:CS}) satisfies a restricted isometry property~(RIP) condition~\cite{candes2005decoding}. However, for general $\{\bm{a}_k\}_{k=1}^N$ (e.g., i.i.d. Gaussian vectors), even if RIP holds for rescaled version of $\bm{A}=\left[\bm{a}_1,\bm{a}_2,\cdots,\bm{a}_{N}\right]$, the matrix $\bm{\A}$ in \eqref{def:calA} can not be shown to satisfy the RIP condition. Instead, we consider the mutual incoherence property (MIP) \cite{donoho2001uncertainty} of $\bm{A}$. We theoretically show that the convergence of IHT can be guaranteed if the mutual coherence of the codebook matrix $\bm{A}$ is small. Moreover, we will provide examples of the codebook matrix $\bm{A}$ that satisfy the MIP condition.\par
Let
\begin{equation}\label{def:mutualcoherence}
       \mu=\max_{\substack{k\neq \ell,\\ k,\ell \in [N] }}\frac{\lv\l \bm{a}_k,\bm{a}_l\r\rv}{\lV\bm{a}_k\rV_2 \lV\bm{a}_l\rV_2},
\end{equation}
be the mutual coherence of $\{\bm{a}_k\}_{k=1}^N$ (or the matrix $\bm{A}$).
\begin{theorem}\label{theorem:IHTconv}
       Assume that $\lV\bm{a}_k\rV_2=\sqrt{D}$, $\forall k\in [N]$. Let $\mu$ be given in \eqref{def:mutualcoherence},  $\eta$ be given in \eqref{problem:l0recovery}, and the sequence $\{\bm{\gamma}_{t}\}_{t\ge 0}$ be defined by IHT in \eqref{iter:IHT}. For some constant step size $\alpha$, if $\{\bm{a}_k\}_{k=1}^N$ satisfy the MIP condition
       $\mu\le\frac{1}{3^{1/4}\sqrt{3K-1}}$,
       then it holds
       \begin{equation*}
              \lV\bm{\gamma}_{t+1}-\bm{\gamma}^\natural\rV_2\le \upsilon\lV\bm{\gamma}_{t}-\bm{\gamma}^\natural\rV_2+\zeta\lV\eta\rV_2,~\forall t\ge 0,
       \end{equation*}
       where $\upsilon\in (0,1)$, $\zeta>0$ are constants.
\end{theorem}
\begin{proof}
       See Appendix~\ref{proof:IHTconv}.
\end{proof}
Various choices of $\{\bm{a}_k\}_{k=1}^{N}$ or $\bm{A}$ are available to ensure small mutual coherence (see {\cite[Chapter 5]{foucart2013}}). Examples go as follows.
\begin{example}\label{example:muofA}
       Let $\{\bm{a}_k\}_{k=1}^{N}$ be i.i.d. random variables drawn uniformly from the sphere of radius $\sqrt{D}$. Then, as long as $8\log (DN)<D<\infty$, we have
       \begin{equation}\label{eq:complexmipforA}
              \mathbb{P}\left[\mu\le \sqrt{\frac{32\log(DN)}{D}}\right]\ge 1-\frac{c_0}{D\sqrt{D}},
       \end{equation}
       where $c_0$ is an universal constant.
\end{example}

\begin{example}\label{example:muofrealA}
       Let $\bm{A}\in \mathbb{R}^{D\times N}$ be Bernoulli random matrix taking values $\pm 1$ with equal probability. Then,  we have
       \begin{equation}\label{eq:complexmipforrealA}
              \mathbb{P}\left[\mu\le 2\sqrt{\frac{\log(N/\varepsilon)}{D}}\right]\ge 1-\varepsilon^2.
       \end{equation}
\end{example}

\noindent Example~\ref{example:muofA} is a direct consequence of \eqref{eq1:lemmaconAK} in Lemma~\ref{lemma:conAk} (see Appendix~\ref{sec:lemmas}), Example~\ref{example:muofrealA} is from~\cite[Exercise 9.1]{foucart2013}. Together with Theorem~\ref{theorem:IHTconv}, \cref{eq:complexmipforA,eq:complexmipforrealA} imply that it theoretically requires  $D\ge \widetilde{\mathcal{O}}(K)$ to guarantee successful recovery of IHT algorithm.


\subsection{Recovery Guarantees of Algorithm~\ref{alg:two_stage}}\label{section:theoryforalg2}
The second stage of Algorithm~\ref{alg:two_stage} (i.e., ML or NNLS) has been analyzed in the previous work~\cite{fengler2021non,haghighatshoar2018improved}, thus we focus on the first stage of the algorithm. By \eqref{eq:equiIHT}, we know that selecting $\hat{K}$ largest components in $\{\Y_k\}_{k=1}^N$ is equivalent to the one-step IHT in~\eqref{iter:IHT}.  Consider that $D$, $M$ are reasonably large and $\hat{K}=K$. The following theorem shows that the estimation of one-step IHT satisfies $\mathrm{supp}\left(\bm{\gamma}_1\right)=S^{\natural}$ with overwhelming probability.

\begin{theorem}\label{theorem:Yrecovery}
       Assume that $\{\bm{a}_k\}_{k=1}^{N}$ are i.i.d. random variables drawn uniformly from the sphere of radius $\sqrt{D}$, $D<\infty$. Let $\{\Y_k\}_{k=1}^N$ be defined in \eqref{def:Yk}, $\hat{S}$ be the indices of $K$ largest components in $\{\Y_k\}_{k=1}^N$. Then we have
       \begin{equation*}
              \mathbb{P}\left(\hat{S}=S^{\natural}\right) \ge 1-\varepsilon-\frac{c_0}{D\sqrt{D}},
       \end{equation*}
       provided that
       \begin{equation}\label{ineq:conditionforM}
              M \ge c_4\left(c_\gamma K+\frac{\sigma^2}{\gamma^{\natural}_{\min}}\right)^2\log \left(\frac{(eD)^2}{\varepsilon}\right)\quad \text{and}\quad K\le c_5\frac{D}{c_\gamma\log(ND)},
       \end{equation}
       where $c_4,c_5$ are universal positive constants and $c_\gamma=\frac{\gamma^{\natural}_{\max}}{\gamma^\natural_{\min}}$.
\end{theorem}
\begin{proof}
       See Appendix~\ref{proof:theoremYrec}.
\end{proof}
From Theorem~\ref{theorem:Yrecovery}, we see that there is $S^{\natural}\subset\mathrm{supp}\left(\bm{\gamma}_1\right)$ for $\hat{K}\ge K$. Nevertheless, when taking the largest $\hat{K}\ge K$ components in $\{\Y_k\}_{k=1}^N$, it is possible to detect $K=\widetilde{\mathcal{O}}(D^2)$ active users. Specially, when $\hat{K}\to N$, the method is reduced to the standard ML estimator, which is able to detect $K=\widetilde{\mathcal{O}}(D^2)$ active users.

Theorem~\ref{theorem:rhoYrecovery} below shows that the estimated set $S^0$ in Algorithm~\ref{alg:two_stage} will contain all the indices in $S^{\natural}$ for the choice $\rho=1$. Thus the first stage of Algorithm~\ref{alg:two_stage} is able to reduce the dimension of the coordinate-wise optimization.

\begin{theorem}\label{theorem:rhoYrecovery}
       Under the same hypothesis in Theorem~\ref{theorem:Yrecovery}. Let $\{\Y_k\}_{k=1}^N$ be defined in \eqref{def:Yk}, $S^0$ be the set given by Algorithm~\ref{alg:two_stage} with $\rho=1$. Then we have
       \begin{equation*}
              \mathbb{P}\left(S^{\natural}\subset S^0\right) \ge 1-\varepsilon-\frac{c_0}{D\sqrt{D}},
       \end{equation*}
       provided that $M$ satisfies inequality in \eqref{ineq:conditionforM} and $K \le \min\left\{\frac{c_6 D}{c_\gamma\log(DN)},\frac{c_7 N}{c_\gamma}\right\}$. Where $c_6,c_7$ are universal positive constants.
\end{theorem}
\begin{proof}
       See Appendix~\ref{proof:rhotheoremYrec}.
\end{proof}

Recall that Theorem~\ref{theorem:Yrecovery} guarantees that  $S^{\natural}$ can be successfully recovered via one-step IHT when $K=\widetilde{\mathcal{O}}(D)$. To ensure $K=\widetilde{\mathcal{O}}(D^2)$ active users can be detected, we need a larger size of $S^0$ by setting the parameter $\rho$. When $\rho$ is set to be around $1$, numerical experiments suggest that the accelerated ML has similar performance compared to ML even for $K>D$. In the next section, simulation results will be provided to validate the performance of the proposed Algorithm~\ref{alg:two_stage}.

\section{Simulation and Discussions}\label{section:experiments}
In this section, we present numerical results to compare the proposed Algorithm~\ref{alg:two_stage} with other existing algorithms. For convenience, we will refer Algorithm \ref{alg:two_stage} as ``AccML''~(accelerated maximum-likelihood) in the remaining of this paper.

Column-normalized sub-sampled Fourier matrices~(for a fair comparison with AMP~\cite{fengler2019sparcs,liu2018massive}) are used for the codebook $\bm{A}$ in~\eqref{problem:matrix}, i.e., $\lV \bm{a}_k\rV_2^2=D$. The per-user SNR~(for active users) is defined as
\begin{equation}\label{eq:SNR}
       \mathrm{SNR}=\frac{\lV \bm{a}_k\rV_2^2g_k^{\natural}\mathbb{E}[\lV\bm{h}_k\rV_2^2]}{\mathbb{E}[\lV\bm{Z}\rV_F^2]} = \frac{Dg_k^{\natural}M}{DM\sigma^2}=\frac{g_k^{\natural}}{\sigma^2},
\end{equation}
where $\mathbb{E}[\lV\bm{h}_k\rV_2^2]=M$ since $\bm{h}_k\sim \mathcal{CN}(0,\bm{I}_M)$.
Then, the per-user $E_b/N_0$~(energy per bit to noise power spectral density ratio) in defined as
\begin{equation}\label{eq:EbN0}
       \frac{E_b}{N_0}=\frac{\mathrm{SNR}}{R}=\frac{LDg_k^{\natural}}{B\sigma^2},
\end{equation}
where $R=\frac{B}{D_{total}}=\frac{B}{L\cdot D}$.

Per-user probability of error~(PUPE) is used as the performance metric. By following~\cref{subsec:std-decoding}, the set of decoded messages is denoted as $\mathcal{T}(\bm{Y})\equiv \mathcal{T}({\mathcal{L}(\bm{Y}_l)|l\in [1:L]})=\left\{\hat{\bm{x}}_1,\hat{\bm{x}}_2,\hat{\bm{x}}_3,\cdots\right\}$ and $|\mathcal{T}(\bm{Y})|=K_r\leq T$. For reference, the set of true messages of $K$ users is denoted by $\mathcal{W}_a=\left\{\bm{x}_1,\bm{x}_2,\cdots,\bm{x}_{K}\right\}$ and $|\mathcal{W}_a|=K$.
A miss detection error is declared if $\bm{x}_k\notin \mathcal{T}(\bm{Y})$ and a false alarm error is declared if $\hat{\bm{x}}_r\in \mathcal{T}(\bm{Y})\setminus\mathcal{W}_a$. Then, PUPE can be calculated as
\begin{equation} \label{eq:PUPE}
       P_e = P_{md} + P_{fa},
\end{equation}
where
\begin{equation*}
       P_{md} = 1-\frac{\mathbb{E}[|\mathcal{T}(\bm{Y})\cap\mathcal{W}_a|]}{K},
       P_{fa} = \frac{\mathbb{E}[|\mathcal{T}(\bm{Y})\setminus\mathcal{W}_a|]}{K_r}
\end{equation*}
are per-user probability of miss detection and false alarm, respectively.

\subsection{SPARC Decoding}

To evaluate the effectiveness of AccML for SPARC decoding, we consider a single section case~($L=1$, $B=L\cdot J=J$) with no tree code involved. \cref{fig:MSE_Scaling_M} demonstrates the PUPE performance of AccML and \cref{fig:MSE_Scaling_J} demonstrates the related computational complexity in terms of processing time in second. One-Step IHT in Algorithm~\ref{alg:oneIHT}, ML and NNLS~\cite{fengler2021non} as well as AMP~\cite{liu2018massive} and OAMP~\cite{ma2017orthogonal} are included for comparison. Moreover, another covariance-based approach call Bandit Sampling ML~(BSML)~\cite[Algorithm 2]{dong2022faster} is added for comparison, where the initialization parameter is given by $B_{BSML}=N/2=2^{J}/2$ and $\epsilon_{BSML}=0.6$~\cite[Section VI]{dong2022faster}.\footnote{The initialization parameter $\bm{\alpha}$ and $\bm{\beta}$ for the Thompson Sampling ML~\cite[Algorithm 3]{dong2022faster} is not given clearly. Considering that its effect on computational complexity reduction is fundamentally similar to the BSML, we omit its simulation results here.} For simplicity, all LSFCs $g_k^{\natural}$ are assumed to be $1$ and $\hat{K}=K$ is known at the receiver. For AccML, $\rho=1$ and $\rho=1.05$ are included to show the effect of $\rho$.

\begin{figure}[tbh]
       \centering
       \captionsetup[subfigure]{justification=centering}
       \begin{subfigure}{0.45\textwidth}
              \includegraphics[width=\textwidth]{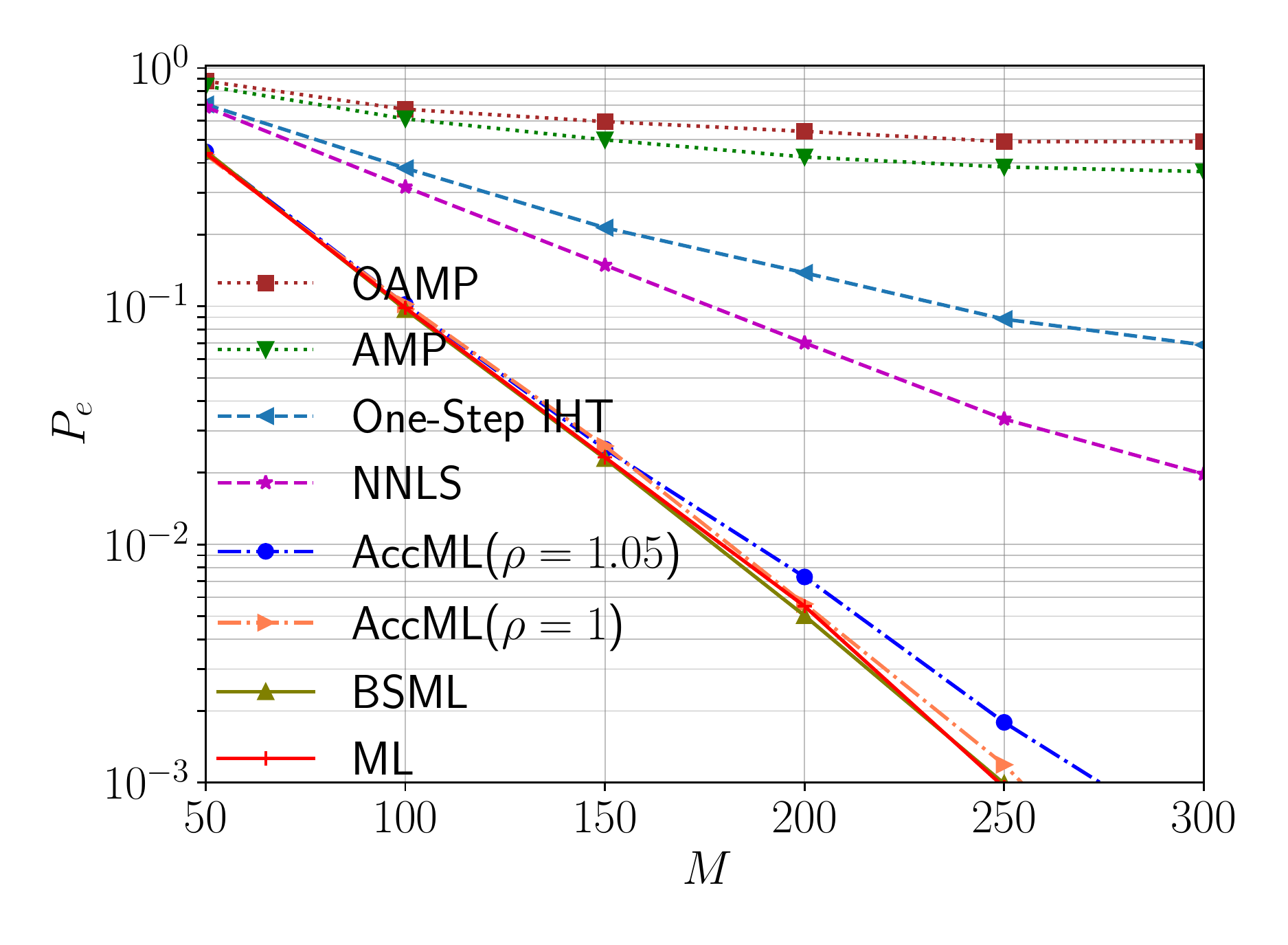}
              \vspace{-1.2cm}
              \caption{\label{fig:MSE_Scaling_M}}
       \end{subfigure}
       \hfill
       \begin{subfigure}{0.45\textwidth}
              \includegraphics[width=\textwidth]{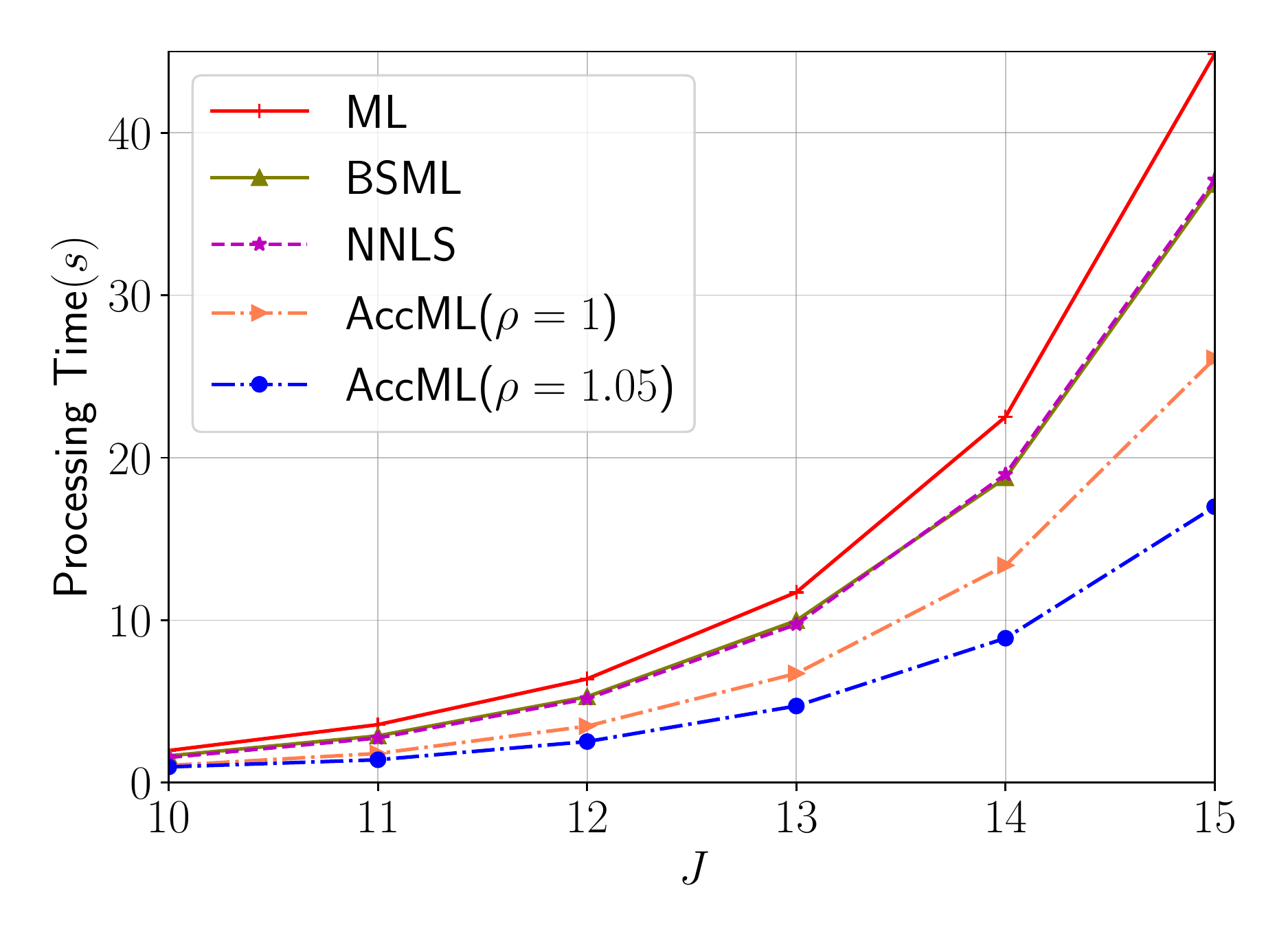}
              \vspace{-1.2cm}
              \caption{\label{fig:MSE_Scaling_J}}
       \end{subfigure}
       \captionsetup{subrefformat=parens}
       \vspace{-0.5cm}
       \caption{\subref{fig:MSE_Scaling_M} PUPE as a function of $M$ where $K=300, D=120, J=12, N=2^J= 4096$ at $\mathrm{SNR}=-10~\mathrm{dB},\frac{E_b}{N_0}=0~\mathrm{dB}$, \subref{fig:MSE_Scaling_J} Processing Time$(s)$ as a function of $J=\log_2{N}$ where $K=300, M=300, D=120$ at $\mathrm{SNR}=-10~\mathrm{dB},\frac{E_b}{N_0}=0~\mathrm{dB}$. AMP, OAMP and One-Step IHT are omitted due to their poor performance in~\cref{fig:MSE_Scaling_M}.}
\end{figure}

As shown in~\cref{fig:MSE_Scaling_M}, both AMP and OAMP give poor PUPE performance for $K>D$. We must emphasize that, according to~\cite{fengler2021non,haghighatshoar2018improved}, up to $K=\widetilde{\mathcal{O}}(D)$ active users can be successfully recovered by AMP and OAMP under the compressive sensing framework. Since the current setup has $K\gg D$ active users, it is beyond the recovery capability of AMP type algorithms. Actually, both AMP and OAMP suffer from instability issue and results in early termination~\cite{fengler2019sparcs}.

When ML is combined with the one-step thresholding, the resulting AccML can achieve the similar PUPE performance as ML and BSML in~\cref{fig:MSE_Scaling_M} but has much lower computational complexity as shown in~\cref{fig:MSE_Scaling_J}. The computational complexity of ML is $\mathcal{O}(N\cdot D^2)$, which increases linearly with $N=2^J$ and thus exponentially with $J$. AccML can reduce the dimension of the problem with the one-step thresholding-based algorithm and the reduction effect increases with $J$ as shown in~\cref{fig:MSE_Scaling_J}.

Comparing $\rho=1$ and $\rho=1.05$ for AccML, a larger $\rho$ means that fewer support indices are selected by the thresholding-based algorithm, which induces larger PUPE as is shown in~\cref{fig:MSE_Scaling_M}. Meanwhile, a larger $\rho$ means that higher dimensions are reduced for ML and thus results in more computational complexity reduction as shown in~\cref{fig:MSE_Scaling_J}.

\subsection{Unsourced Random Access}
The performance of AccML is further evaluated for URA. For a fair comparison, the setup in~\cite[Fig.~8(a)]{fengler2021non} is selected.
However, the thresholding parameter for the decision of support is not provided in details in~\cite{fengler2021non}. Equivalently, we select the indices of largest $T$ components of $\hat{\bm{\gamma}}$ as the support. For simplicity and fairness, we set $T=K$ in this example. Finally, for AccML initialization, we set $\rho=1$.


\cref{fig:Pe_EbN0_Large} compares the PUPE performance of ML and AccML. AccML can achieve similar PUPE performance as ML for $E_b/N_0\le0~\mathrm{dB}$ but the gap becomes larger as $E_b/N_0$ increases. However, for URA in mMTC, the PUPE of practical interest is around $0.05$ and the active users are generally working at $E_b/N_0\le0~\mathrm{dB}$ due to power limitation. In short, AccML performs well for this setup.

Consider that the BS has $M=300$ antennas. This setup could support URA for $K=300$ active users with ML or AccML algorithm. However, the number of antennas available at BS is around $64$ to $128$ for most current massive MIMO systems. Therefore, the maximum active user can be supported is also around $100$. In the next subsection, we consider a more practical setup for URA.

\begin{figure}[tbh]
	\centering
	\captionsetup[subfigure]{justification=centering}
	\begin{subfigure}{0.45\textwidth}
		\includegraphics[width=\textwidth]{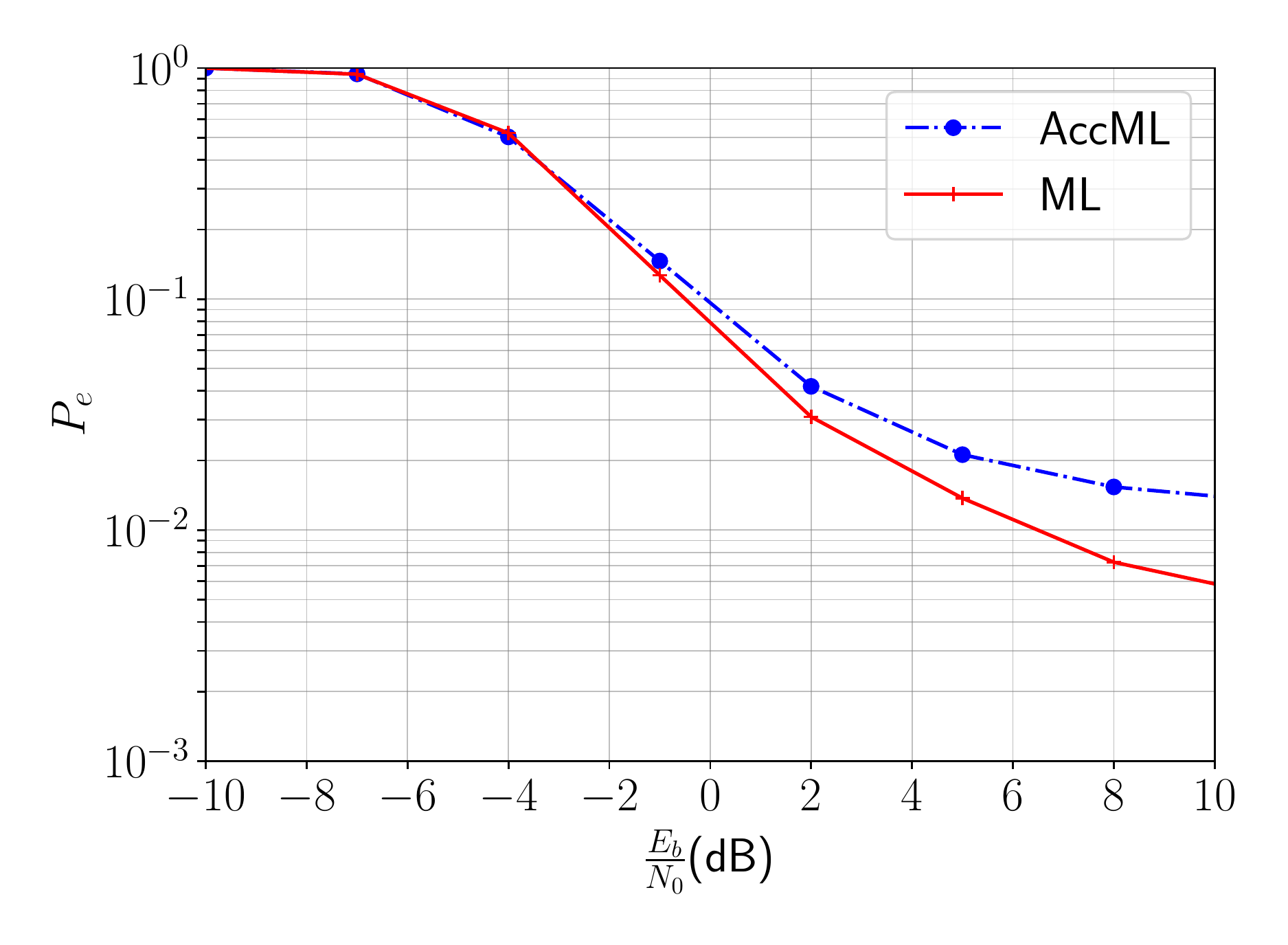}
		\vspace{-1.2cm}
		\caption{\label{fig:Pe_EbN0_Large}}
	\end{subfigure}
	\hfill
	\begin{subfigure}{0.45\textwidth}
		\includegraphics[width=\textwidth]{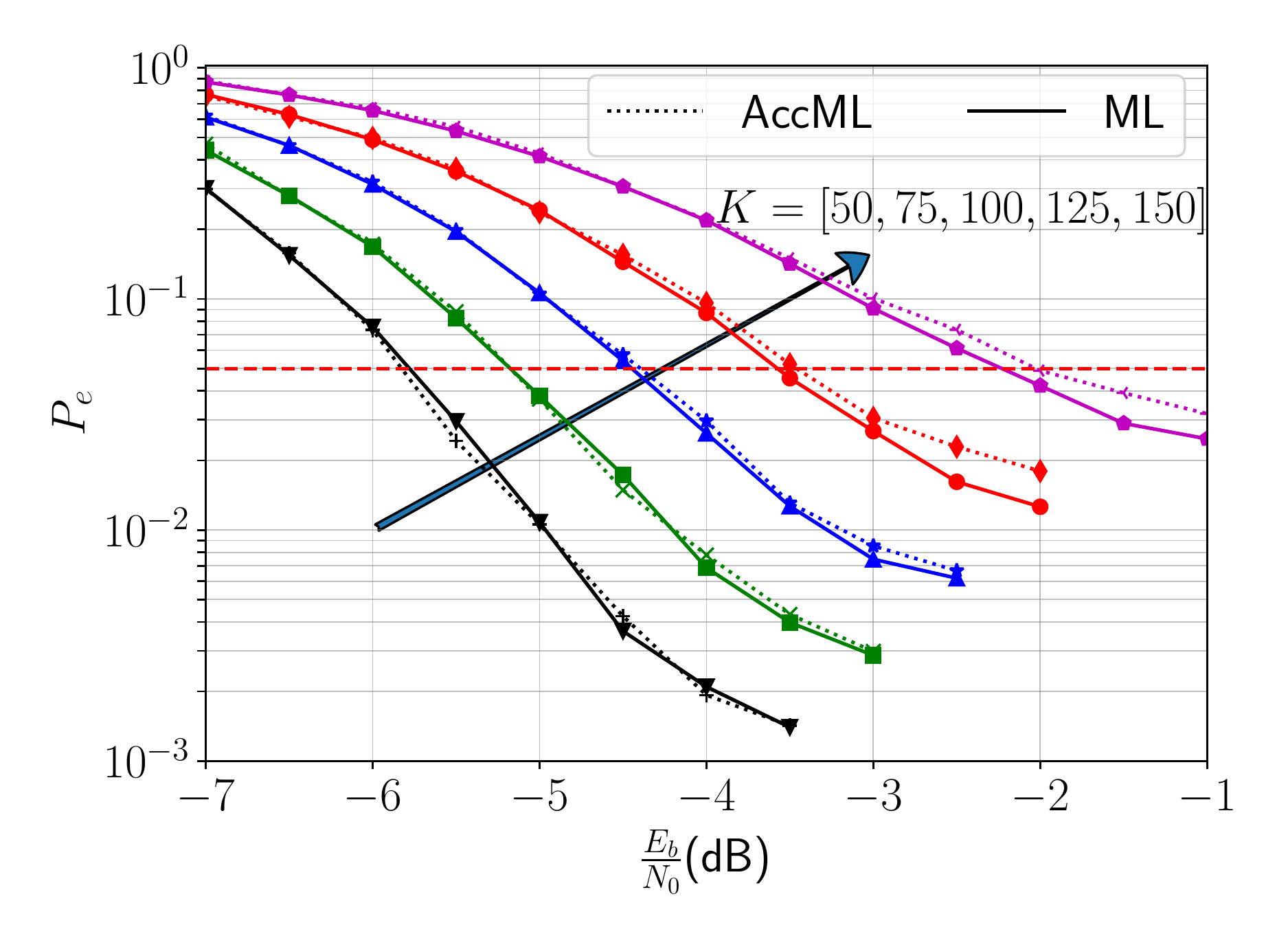}
		\vspace{-1.2cm}
		\caption{\label{fig:Pe_EbN0_AccML}}
	\end{subfigure}
	\captionsetup{subrefformat=parens}
	\vspace{-0.5cm}
	\caption{PUPE as a function of $E_b/N_0~(\mathrm{dB})$. \subref{fig:Pe_EbN0_Large} $K=300, M=300, D=100, D_{total}=3200, B=96, L=32, J=12$. The parity bits across $L=32$ sections are $0,9,\cdots,9,12,12,12$, \subref{fig:Pe_EbN0_AccML} $K=\left[50,75,100,125,150\right], M=64, D=120, D_{total}=1440, B=50, L=12, J=12$. The parity bits across $L=12$ sections are $0,7,8,\cdots,8,11,12$. The red dash line marks the practically interested PUPE of $0.05$. }
\end{figure}

\subsection{Practical MIMO-OFDM Scenarios}
We consider physical uplink shared channel~(PUSCH) in a MIMO orthogonal frequency division multiplexing~(OFDM) system. A total of $12$ resource blocks~(RB) are allocated for transmission and each RB contains $12$ sub-carriers with sub-carrier spacing of $15\;\mathrm{kHz}$~\cite[Clause 4]{TS38.211}. Each sub-carrier contains $14$ channel uses (or OFDM symbols), where the last $10$ OFDM symbols are used for PUSCH and the remaining $4$ OFDM symbols are reserved for other usages such as control signal~\cite[Clause 11]{TS38.213}. Each section occupies a single RB of $120$ channel uses. Channel is assumed to stay static within each RB and may change randomly across different RB.

In summary, a total of $D_{total}=1440$ channel uses are used for URA, which are divided into $L=12$ sections~(corresponds to $12$ RB) with $D=120$ channel uses each. Consider that $K=[50,75,100,125,150]$ active user are supported by BS with $M=64$ antennas. Each user transmit $B=50$ information bits across $L=12$ sections. Each section is augmented to size $J=12$ with the additional parity bits, which are allocated as $0,7,8,\cdots,8,11,12$. For the SPARC decoding, we set $T$ as $K+K_\Delta$ with $K_\Delta=50$ and $\rho=1.05$ for AccML initialization. For simplicity, the details of standard OFDM operations involving adding/stripping off cyclic prefix and the IFFT/FFT operation are omitted.

The solid and dot lines in~\cref{fig:Pe_EbN0_AccML} show the PUPE performance versus ${E_b}/{N_0}$ for ML and AccML, respectively. Overall, AccML achieves the similar performance as ML. When the number of the active users increases, the performance of AccML degrades slightly. Intuitively, the hard-thresholding of Step~\ref{step:threshold} in Algorithm~\ref{alg:two_stage} may miss some non-zero entries in $\bm{\gamma}^{\natural}$, which results in lager PUPE.

\begin{figure}[tbh]
       \centering
       \captionsetup[subfigure]{justification=centering}
       \begin{subfigure}{0.45\textwidth}
              \centering
              \includegraphics[width=\textwidth]{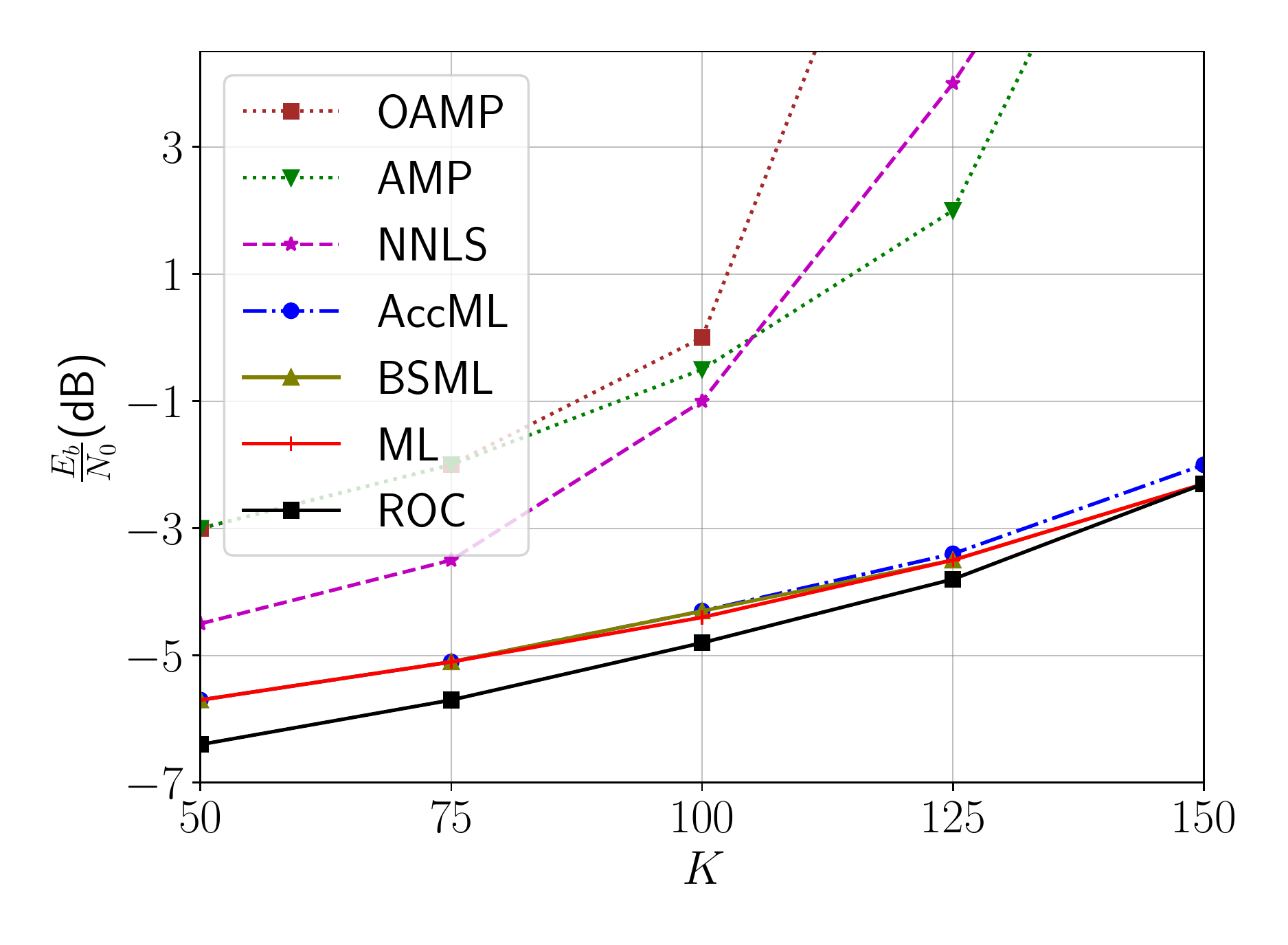}
              \vspace{-1.2cm}
              \caption{\label{fig:EbN0_required}}
       \end{subfigure}
       \hfill
       \begin{subfigure}{0.45\textwidth}
              \centering
              \includegraphics[width=\textwidth]{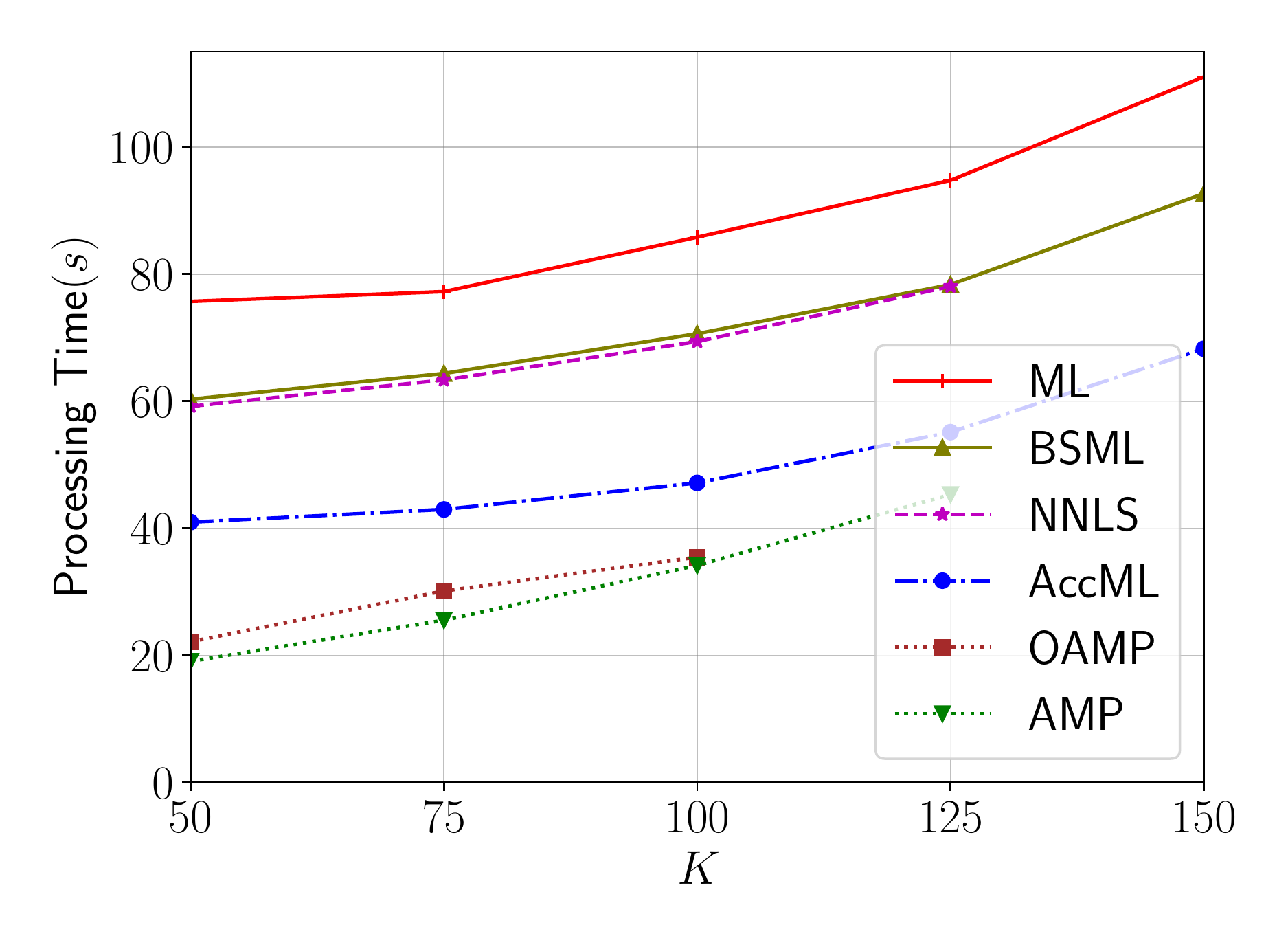}
              \vspace{-1.2cm}
              \caption{\label{fig:Time_K}}
       \end{subfigure}
       \captionsetup{subrefformat=parens}
       \vspace{-0.5cm}
       \caption{\subref{fig:EbN0_required} The minimal ${E_b}/{N_0}\;\mathrm{(dB)}$ to guarantee PUPE to be less than $0.05$ as a function of $K$, \subref{fig:Time_K} Processing Time$(s)$ as a function of $K$. We exclude the results of large $K_a$ for NNLS, AMP and OAMP when they fail to converge at reasonable low ${E_b}/{N_0}$ as shown in~\cref{fig:EbN0_required}.}
\end{figure}

\cref{fig:EbN0_required} shows the required ${E_b}/{N_0}$ to achieve $\mathrm{PUPE}\leq0.05$ versus different number of active users. For comparison, AMP, OAMP, BSML and NNLS are also included. Finally, the asymptotic performance analysis based on the receiver operating characteristic~(ROC) of ML~\cite{chen2019cov,poor2013introduction} is added for reference, which is denoted as ``ROC''. As shown in~\cref{fig:EbN0_required}, AccML achieves the similar performance as ML and BSML and significantly outperforms NNLS, AMP and OAMP. In particular, for $K\geq 100$, NNLS, AMP and OAMP deteriorate dramatically but AccML still performs near optimally. It should be highlighted that both AMP and OAMP work under extremely poor condition for $K \geq 100$, which results in early termination. Overall, in terms of both ${E_b}/{N_0}$ requirement and the scaling law in $K$ and $D$, the AccML outperforms AMP and OAMP significantly.

The processing time of both SPARC and tree decoding is shown in~\cref{fig:Time_K}. For SPARC decoding, AccML could reduce the computational complexity of the ML stage from $\mathcal{O}(D^2\cdot N)$ as $\mathcal{O}(D^2\cdot K_{sel})$ where $K_{sel}$ is the number of indices in $S^{0}$. For comparison, the BSML alternates between Bandit Sampling of computational complexity of $\mathcal{O}(D^2\cdot B_{BSML})$ and standard ML of of computational complexity of $\mathcal{O}(D^2\cdot N)$~\cite{dong2022faster}. The processing time of AccML is much smaller than ML, BSML and NNLS. In particular, when $K>100$, AccML only consumes slightly larger time than AMP but achieve the near optimal performance as ML and BSML. It demonstrates the advantages of the proposed two-stage AccML over all existing methods.

\section{Conclusion}\label{section:conclusion}
In this paper, we propose a two-stage method for SPARC decoding that combines a one-step thresholding-based algorithm for dimension reduction and ML for refinement. Due to the dimension reduced in the first stage, the complexity of ML is significantly reduced while the overall performance still being maintained. Adequate simulation results are provided to validate the near-optimal performance and the low computational complexity. In addition, theoretical analyses are given for both the covariance-based sparse recovery model and the thresholding algorithm. When convex relaxation is considered, the upper bound for the number of active users is given and proven to be optimal. For the one-step thresholding-based algorithm, theoretical analyses show that it is highly possible to contain all the active users when $M$ and $D$ are reasonably large. For future works, the extension of the proposed algorithm and theoretical analyses to more practical spatially correlated channel models seems promising.


%
\appendix
\appendixpage
\addappheadtotoc
\section{Key Lemmas}\label{sec:lemmas}

\begin{lemma}\label{lemma:expak4}
       Assume that $\bm{a}_k\in\mathbb{C}^D$, $k\in[N]$ are i.i.d random variables drawn uniformly from the sphere of radius $\sqrt{D}$. Then for all $k\in [N]$, $i,j\in [D]$ there is
       \begin{equation*}
              \mathrm{E}[\lv a_{k,i}\rv^2\lv a_{k,j}\rv^2]
              = \frac{D}{D+1},\quad \forall i\neq j,
       \end{equation*}
       and
       \begin{equation*}
              \mathrm{E}[\lv a_{k,i}\rv^4]
              = \frac{2D}{D+1},\quad \forall i\in [D].
       \end{equation*}
\end{lemma}

\begin{proof}
       For all $k\in [N]$, define the real-valued vector
       \begin{equation*}
              \bm{a}_k^{\mathbb{R}}:=[\mathrm{Re}(\bm{a}_k),\mathrm{Im}(\bm{a}_k)]\in\mathbb{R}^{2D},
       \end{equation*}
       where $\mathrm{Re}(\bm{a}_k)$ and $\mathrm{Im}(\bm{a}_k)$ are the real part and imaginary part of $\bm{a}_k$, respectively. Then $\{\bm{a}_k^{\mathbb{R}}\}_{k=1}^N$ are also i.i.d. random variables drawn uniformly from the sphere of radius $\sqrt{D}$. And we have
       \begin{align}\label{eq:expakiakj1}
              \mathrm{E}[\lv a_{k,i}\rv^2\lv a_{k,j}\rv^2]
              = & ~\mathrm{E}\Big[\left(\lv\mathrm{Re}(a_{k,i})\rv^2+\lv\mathrm{Im}(a_{k,i})\rv^2\right) \left(\lv\mathrm{Re}(a_{k,j})\rv^2+\lv\mathrm{Im}(a_{k,j})\rv^2\right)\Big] \nonumber \\
              = & ~4\mathrm{E}\left[\lv a_{k,i}^{\mathbb{R}}\rv^2\lv a_{k,j}^{\mathbb{R}}\rv^2\right].
       \end{align}
       As $\frac{\bm{a}_k^{\mathbb{R}}}{\sqrt{D}}$ are i.i.d. random variables drawn uniformly from the unit sphere, we rewrite the vector $\frac{\bm{a}_k^{\mathbb{R}}}{\sqrt{D}} := \frac{\bm{g}_k}{\lV\bm{g}_k\rV_2}$, where $\bm{g}_k\in \mathbb{R}^{2D}$ with each component of $\bm{g}_k$ be i.i.d. standard Gaussian. Then
       \begin{align*}
              \frac{1}{D^2}\mathrm{E}\left[\lv a_{k,i}^{\mathbb{R}}\rv^2\lv a_{k,j}^{\mathbb{R}}\rv^2\right]
              = & ~\mathrm{E}\left[\frac{\lv g_{k,i}\rv^2\lv g_{k,j}\rv^2}{\lV\bm{g}_k\rV_2^4}\right]
              =\mathrm{E}\left[\lv g_{k,i}\rv^2\lv g_{k,j}\rv^2\int_{0}^{+\infty}te^{-t\lV\bm{g}_k\rV_2^2}dt\right]                                                                             \\
              = & ~\mathrm{E}\left[\int_{0}^{+\infty}\lv g_{k,i}\rv^2 e^{-t\lv g_{k,i}\rv^2}\lv g_{k,j}\rv^2 e^{-t\lv g_{k,j}\rv^2} \prod\limits_{s\neq i,j} e^{-t\lv g_{k,s}\rv^2}  tdt\right] \\
              = & ~\int_{0}^{+\infty}[\tilde{f}'(t)]^2[\tilde{f}(t)]^{2D-2}tdt,
       \end{align*}
       where
       $\tilde{f}(t)=\mathrm{E}\left[e^{-t\lv g_{k,i}\rv^2}\right]=\frac{1}{\sqrt{1+2t}}$, and $\tilde{f}'(t)=-(1+2t)^{3/2}$. Thus,
       \begin{align*}
              \mathrm{E}\left[\lv a_{k,i}^{\mathbb{R}}\rv^2\lv a_{k,j}^{\mathbb{R}}\rv^2\right]
              =D^2   \int_{0}^{+\infty}[\tilde{f}'(t)]^2[\tilde{f}(t)]^{2D-2}tdt
              =D^2\int_{0}^{+\infty}(1+2t)^{-D-2}tdt=\frac{D}{4D+4}.
       \end{align*}
       Therefore, by \eqref{eq:expakiakj1} we have
       \begin{equation*}
              \mathrm{E}[\lv a_{k,i}\rv^2\lv a_{k,j}\rv^2]
              = 4\mathrm{E}\left[\lv a_{k,i}^{\mathbb{R}}\rv^2\lv a_{k,j}^{\mathbb{R}}\rv^2\right]  = \frac{D}{D+1}.
       \end{equation*}
       Now we compute $\mathrm{E}[\lv a_{k,i}\rv^4]$. By $\lV\bm{a}_k\rV_2^4=D^2$, we have
       \begin{align*}
              D^2 =\sum_{i\neq j}\mathrm{E}[\lv a_{k,i}\rv^2\lv a_{k,j}\rv^2]+\sum_{i=1}^D \mathrm{E}[\lv a_{k,i}\rv^4]
              =D(D-1)\mathrm{E}[\lv a_{k,i}\rv^2\lv a_{k,j}\rv^2]
              +D \mathrm{E}[\lv a_{k,i}\rv^4],
       \end{align*}
       for all $i\in [D]$, and $i\neq j$. Thus we have
       \begin{equation*}
              \mathrm{E}[\lv a_{k,i}\rv^4]
              =D- (D-1)\mathrm{E}[\lv a_{k,i}\rv^2\lv a_{k,j}\rv^2]
              =\frac{2D}{D+1}.
       \end{equation*}
\end{proof}
\begin{lemma}[{[\citenum{cai2013distributions}, Proposition 5]}]\label{ineq:angle}
       Let the vectors $\bm{x}, \bm{y}\in \mathbb{R}^{n}$ be sampled independently and uniformly from the unit sphere $\mathbb{S}^{n-1}$, and $\theta$ be the angle between $\bm{x}$ and $\bm{y}$. Then,
       \begin{equation*}
              \mathbb{P}\left(\lv \theta-\frac{\pi}{2}\rv \ge \varepsilon\right)\le \tilde{c}_0\sqrt{n}\left(\cos \varepsilon\right)^{n-2},
       \end{equation*}
       for all $n\ge 2$ and $\varepsilon\in(0,\frac{\pi}{2})$, where $\tilde{c}_0$ is an universal constant.
\end{lemma}
\begin{lemma}\label{lemma:conAk}
       Assume that $\{\bm{a}_k\}_{k=1}^{N}$ are i.i.d. random variables drawn uniformly from the sphere of radius $\sqrt{D}$.  For any non-negative $\bm{\gamma}^{\natural}$ supported on $S^{\natural}$ and any  $\varepsilon_0>0$, if provided $8\log(DN)<D<\infty$, then
        it holds with probability at least $1-\frac{c_0}{D\sqrt{D}}$ that
       \begin{equation}\label{eq1:lemmaconAK}
              \lv\frac{1}{D}\bm{a}_k^*\bm{a}_\ell\rv \le \sqrt{\frac{32\log (DN)}{D}},\quad \forall k\neq \ell ,
       \end{equation}
       and
       \begin{equation}\label{eq2:lemmaconAK}
              \sum_{\ell\in S^{\natural}}\gamma^{\natural}_{\ell}\lv\sum_{i=1}^D\overline{a}_{k,i} a_{\ell,i}\rv^2 \le \sum_{\ell\in S^{\natural}}32\gamma^{\natural}_{\ell} D\log (DN) ,
       \end{equation}
       for any $ k\notin S^{\natural}$. $c_0$ is a universal constant.
\end{lemma}

\begin{proof}

       Recall that $\bm{a}_k^{\mathbb{R}}=[\mathrm{Re}(\bm{a}_k),\mathrm{Im}(\bm{a}_k)]$. We also denote
       \begin{align}\label{def:akR}
              \hat{\bm{a}}_k^{\mathbb{R}}:=[\mathrm{Im}(\bm{a}_k),-\mathrm{Re}(\bm{a}_k)]\in\mathbb{R}^{2D}~k\in [N].
       \end{align}
       Then by definition we have
       \begin{align}\label{eq:akal}
              \lv\frac{1}{D}\bm{a}_k^*\bm{a}_\ell\rv
              = & ~\Big\lvert\frac{1}{D}[\mathrm{Re}(\bm{a}_k)\cdot\mathrm{Re}(\bm{a}_\ell)   +\mathrm{Im}(\bm{a}_k)\cdot\mathrm{Im}(\bm{a}_\ell) -i\cdot\left(\mathrm{Im}(\bm{a}_k)\cdot\mathrm{Re}(\bm{a}_\ell)-\mathrm{Re}(\bm{a}_k)\cdot \mathrm{Im}(\bm{a}_\ell)\right) ]\Big\rvert \nonumber \\
              = & ~\frac{1}{D}\lv\left(\bm{a}_k^{\mathbb{R}}\right)^T\bm{a}_\ell^{\mathbb{R}}-\left(\hat{\bm{a}}_k^{\mathbb{R}}\right)^T\bm{a}_\ell^{\mathbb{R}}\rv.
       \end{align}
       Notice that for $k\neq \ell$, $\frac{\bm{a}_k^{\mathbb{R}}}{\sqrt{D}}\in \mathbb{R}^{2D}$ and $\frac{\bm{a}_\ell^{\mathbb{R}}}{\sqrt{D}}\in\mathbb{R}^{2D}$ are sampled independently and uniformly from the unit sphere $\mathbb{S}^{2D-1}$. Let $\theta_{k,\ell}$ be the angle between $\frac{\bm{a}_k^{\mathbb{R}}}{\sqrt{D}}$ and $\frac{\bm{a}_\ell^{\mathbb{R}}}{\sqrt{D}}$. Then, for some given $k$ and $\ell$ such that  $k\neq \ell$,  by Lemma~\ref{ineq:angle} we have
       \begin{equation*}
              \mathbb{P}\left(\lv\theta_{k,\ell}-\frac{\pi}{2}\rv\ge \tilde{\varepsilon}\right)\le \tilde{c}_1\sqrt{2D}\left(\cos \tilde{\varepsilon}\right)^{2D-2}
       \end{equation*}
       for any $\tilde{\varepsilon}\in(0,\frac{\pi}{2})$ and some universal constant $\tilde{c}_1$. Thus, for any $k\neq \ell$, by a union bound we have
       \begin{equation*}
              \mathbb{P}\left(\lv\theta_{k,\ell}-\frac{\pi}{2}\rv\ge \tilde{\varepsilon}\right)
              \le \tilde{c}_1 \tbinom{N}{2}\sqrt{2D} \left(\cos \tilde{\varepsilon}\right)^{2D-2}.
       \end{equation*}
       For any $\tilde{\varepsilon}\in(0,1)$, we have $\cos \tilde{\varepsilon}\le 1-\frac{\tilde{\varepsilon}^2}{2}+\frac{\tilde{\varepsilon}^4}{24}\le 1-\frac{\tilde{\varepsilon}^2}{4}$. Then, as long as $\frac{8\log (DN)}{D}<1$ and $D<\infty$, we have for some universal constant $c_0$,
       \begin{equation*}
              \mathbb{P}\left(\lv\theta_{k,\ell}-\frac{\pi}{2}\rv\ge \sqrt{\frac{8\log (DN)}{D}}\right)\le \frac{c_0}{2 D\sqrt{D}},~\forall k\neq \ell.
       \end{equation*}
       Therefore, with probability at least $1-\frac{c_0}{2 D\sqrt{D}}$ we have
       \begin{align*}
              \lv\cos\theta_{k,\ell}\rv
              \le \lv\cos\left(\frac{\pi}{2}+\sqrt{\frac{8\log (DN)}{D}}\right) \rv
              \le\lv\sin\left(\sqrt{\frac{8\log (DN)}{D}}\right)\rv\le \sqrt{\frac{8\log (DN)}{D}},\quad \forall k\neq \ell.
       \end{align*}
       By the definition of $\theta_{k,\ell}$, we then have with probability at least $1-\frac{c_0}{2D\sqrt{D}}$ it holds
       \begin{equation*}
              \lv\frac{\left(\bm{a}_k^{\mathbb{R}}\right)^T\bm{a}_\ell^{\mathbb{R}}}{D}\rv=\lv\cos\theta_{k,\ell}\rv \le \sqrt{\frac{8\log (DN)}{D}},~\forall k\neq \ell.
       \end{equation*}
       Similarly, with probability at least $1-\frac{c_0}{2D\sqrt{D}}$ we also have
       \begin{equation*}
              0\le \lv\frac{\left(\hat{\bm{a}}_k^{\mathbb{R}}\right)^T\bm{a}_\ell^{\mathbb{R}}}{D}\rv\le \sqrt{\frac{8\log (DN)}{D}},~\forall k\neq\ell.
       \end{equation*}
       Then, together with \eqref{eq:akal}, it yields
       \begin{align*}
              \lv\bm{a}_k^*\bm{a}_\ell\rv & \le \lv\left(\bm{a}_k^{\mathbb{R}}\right)^T\bm{a}_\ell^{\mathbb{R}}\rv+\lv\left(\hat{\bm{a}}_k^{\mathbb{R}}\right)^T\bm{a}_\ell^{\mathbb{R}}\rv
              \le 4\sqrt{2D\log (DN)}
       \end{align*}
       with probability at least $1-\frac{c_0}{D\sqrt{D}}$. Therefore, for $k\notin S^{\natural}$ and $\ell\in S^{\natural}$ we must have $k\neq \ell$ and thus
       \begin{align*}
              \sum_{\ell\in S^{\natural}}\gamma^{\natural}_{\ell}\lv\sum_{i=1}^D\overline{a}_{k,i} a_{\ell,i}\rv^2 & =\sum_{\ell\in S^{\natural}}\gamma^{\natural}_{\ell}D^2\lv\frac{1}{D}\bm{a}_k^*\bm{a}_\ell\rv^2 \le \sum_{\ell\in S^{\natural}}32\gamma^{\natural}_{\ell}D\log (DN).
       \end{align*}
\end{proof}

\begin{lemma}\label{lemma:MIPinA}
       Let $\mu$ be the mutual coherence of $\{\bm{a}_k\}_{k=1}^{N}$ defined in \eqref{def:mutualcoherence}, and $\bm{\A}$ be defined in \eqref{def:calA} whose columns are denoted by $\{\bm{\A_k}\}_{k=1}^N$. Then, the coherence of $\bm{\A}$, i.e.,
       \begin{equation*}
              \hat{\mu}=\max_{k\neq \ell}\frac{\lv\l \bm{\A}_k,\bm{\A}_l\r\rv}{\lV\bm{\A}_k\rV_2 \lV\bm{\A}_l\rV_2}
       \end{equation*}
       satisfies $ \hat{\mu}=\mu^2$.
\end{lemma}
\begin{proof}
       By the definition of $\bm{\A}$ and $\bm{\A_k}$, for all $k,\ell \in [N]$ we have
       \begin{align*}
              \lv\l \bm{\A}_k,\bm{\A}_l\r\rv
               & =\lv\l \mathrm{vec}(\bm{a}_k\bm{a}_k^*),\mathrm{vec}(\bm{a}_\ell\bm{a}_\ell^*)\r\rv
              =\lv\sum_{1\le i,j\le D}\overline{a_{k,i}a_{k,j}}a_{\ell,i}a_{\ell,j}\rv                                                         \\
               & =\lv\left(\sum_{1\le i\le D}\overline{a}_{k,i}a_{\ell,i}\right)\left(\sum_{1\le j\le D}\overline{a}_{k,j}a_{\ell,j}\right)\rv \\
               & = \lv\l \bm{a}_k,\bm{a}_l\r\rv^2.
       \end{align*}
       By taking $k=\ell$ in the above equation, we also have $\lV\bm{\A}_k\rV_2=\lV\bm{a}_k\rV_2^2$, $\forall k\in [N]$. Thus, by the definition of $\mu$ and $\hat{\mu}$, we have  $\hat{\mu}=\mu^2$.
\end{proof}

\begin{lemma}[{[\citenum{fengler2021non}, Theorem 11]}]\label{bound:Delta}
       Let $\bm{\Delta}=\bm{\widehat{\Sigma}_y}-\bm{\Sigma_y}$ where $\bm{\widehat{\Sigma}_y}$ and $\bm{\Sigma_y}$ are given by \eqref{Y:truesigma} and \eqref{Y:samplesigma} respectively. Let $\varepsilon>0$ and $\lV\bm{a}_k\rV^2_2=D$, $k\in [N]$, then for some universal constant $c$, it holds that
       \begin{equation*}
              \lV \bm{\Delta}\rV_F\le \frac{D\left(\lV\bm{\gamma}^{\natural}\rV_1+\sigma^2\right)}{\sqrt{M}}\sqrt{c\log \left(\frac{(eD)^2}{\varepsilon}\right)}
       \end{equation*}
       with probability at least $1-\varepsilon$ provided $ M\ge c\log \left(\frac{2(eD)^2}{\varepsilon}\right)$.

\end{lemma}

\section{Proof of Proposition~\ref{prop:expZk}}\label{proof:propexpZk}
\begin{proof}
       By the definition of $\Y_k$ and $\bm{\widehat{\Sigma}_y}$, we have
       \begin{align}\label{eq:formYk}
              \Y_k
               & =\frac{1}{M}\sum_{i,j\in [D]}\overline{a}_{k,i}a_{k,j} \bm{Y}(i,:)\bm{Y}(j,:)^*\nonumber \\
               & =\frac{1}{M}\sum_{ i,j \in [D]}
              \Bigg[\overline{a}_{k,i}a_{k,j}\Bigg(\sum_{\ell\in [N]}\sqrt{\gamma^{\natural}_{\ell}}a_{\ell,i}\bm{h}^*_{\ell}+\bm{z}_i^* \Bigg) \Bigg(\sum_{m\in [N]}\sqrt{\gamma^{\natural}_{m}}\overline{a}_{m,j}\bm{h}_{m}+\bm{z}_j \Bigg)
              \Biggl] .
       \end{align}
       The assumption is that $\{\bm{a}_k\}_{k=1}^N$ are independent from channel vectors and AWGN. Then, by the independence of $\{\bm{h}_k\}_{k=1}^N$, $\{\bm{z}_j\}_{j=1}^D$ and $\{\bm{a}_k\}_{k=1}^N$, we have
       \begin{align*}
              \mathrm{E}[\Y_k]
              = & ~\frac{1}{M}\sum_{ i,j \in [D]}\sum_{ \ell,m\in [N]}\mathrm{E}[\overline{a}_{k,i}a_{k,j}a_{\ell,i}\overline{a}_{m,j}]\sqrt{\gamma^{\natural}_{\ell}\gamma^{\natural}_{m}}\mathrm{E}[\bm{h}^*_{\ell}\bm{h}_{m}]+\frac{1}{M}\sum_{i,j\in [D]}\mathrm{E}[\overline{a}_{k,i}a_{k,j}]\mathrm{E}[\bm{z}_i^*\bm{z}_j] \\
              = & ~\sum_{i,j\in [D]}\sum_{\ell\in [N]}\mathrm{E}[\overline{a}_{k,i}a_{k,j}a_{\ell,i}\overline{a}_{\ell,j}]\gamma^{\natural}_{\ell}
              +\sum_{i,j\in [D]}\mathrm{E}[\overline{a}_{k,i}a_{k,j}]\delta_{i,j}\sigma^2  \\
              = & ~\sum_{i=1}^D\sum_{\ell\in S^{\natural} }\mathrm{E}[\lv a_{k,i}\rv^2 \lv a_{\ell,i}\rv^2]\gamma^{\natural}_{\ell} +\sum_{i\neq j}\sum_{\ell\in S^{\natural} }\mathrm{E}[\overline{a}_{k,i}a_{k,j}a_{\ell,i}\overline{a}_{\ell,j}]\gamma^{\natural}_{\ell} +\sum_{i=1}^D \mathrm{E}[\lv a_{k,i}\rv^2]\sigma^2,
       \end{align*}
       where the last equation follows from the fact that $\bm{\gamma}^\natural$ is supported on $S^{\natural}$.
       Assume that $\{\bm{a}_k\}_{k=1}^N$ are i.i.d. random variables drawn uniformly from the sphere of radius $\sqrt{D}$. Then for $k\in S^{\natural}$, there is
       \begin{align*}
              \mathrm{E}[\Y_k]
              = & \sum_{i=1}^D \mathrm{E}[\lv a_{k,i}\rv^4]\gamma^{\natural}_{k} +\sum_{i=1}^D\sum_{\ell\in S^{\natural},\ell\neq k }\mathrm{E}[\lv a_{k,i}\rv^2 \lv a_{\ell,i}\rv^2]\gamma^{\natural}_{\ell} +\sum_{i\neq j}\mathrm{E}[\lv a_{k,i}\rv^2 \lv a_{k,j}\rv^2]\gamma^{\natural}_{k}+\sum_{i=1}^D \mathrm{E}[\lv a_{k,i}\rv^2]\sigma^2 \\
              = & ~\sum_{i=1}^D\frac{2D}{D+1}\gamma^{\natural}_{k}+\sum_{i=1}^D\sum_{\ell\in S^{\natural},\ell\neq k }\frac{D}{D+1}\gamma^{\natural}_{\ell} +\sum_{i\neq j}\frac{D}{D+1}\gamma^{\natural}_{k}+D\sigma^2\\
= & ~\sum_{\ell\in S^{\natural} }\frac{D^2}{D+1}\gamma^{\natural}_{\ell}+D^2\gamma^{\natural}_{k}+D\sigma^2,
       \end{align*}
       where the first equation follows from the fact that
       \begin{equation*}
              \mathrm{E}[\lv a_{k,i}\rv^2]=1,~ \mathrm{E}[a_{k,i}a_{k,j}]=0,\quad\forall i\neq j,~\forall k\in [N],
       \end{equation*}
       the second equation follows from Lemma~\ref{lemma:expak4}.  Also, for $k\notin S^{\natural}$, we know $k\neq \ell$, $\forall \ell\in S^{\natural}$, and thus by Lemma~\ref{lemma:expak4} we have
       \begin{align*}
              \mathrm{E}[\Y_k] = & ~
              \sum_{i=1}^D\sum_{\ell\in S^{\natural} }\mathrm{E}[\lv a_{k,i}\rv^2 \lv a_{\ell,i}\rv^2]\gamma^{\natural}_{\ell} +\sum_{i\neq j}\sum_{\ell\in S^{\natural}}\mathrm{E}[\overline{a}_{k,i}a_{k,j}a_{\ell,i}\overline{a}_{\ell,j}]\gamma^{\natural}_{\ell}+\sigma^2 \\
              =                  & ~\sum_{\ell\in S^{\natural} }\frac{D^2}{D+1}\gamma^{\natural}_{\ell}+D\sigma^2.
       \end{align*}
       Thus we complete the proof.
\end{proof}

\section{Proof of Theorem~\ref{theory:l1norm}}\label{proof:theoreml1rec}
The proof of the theorem relies on the so-called robust null space property (NSP) of  $\bm{\A}$ and a substantial lemma in~\cite{foucart2013}. The robust NSP is described in the following definition.
\begin{definition}(Robust NSP,{\cite[Definition 4.17]{foucart2013}})
       The matrix $\bm{\A}\in \mathbb{C}^{D^2\times N}$ is said to satisfy the robust $\ell_2$-NSP of order $s$ with constants $0<\beta<1$ and $\tau>0$ if
       \begin{equation*}
              \lV\bm{v}\rV_1\le\beta\lV\bm{v}_{\overline{\mathcal{S}}}\rV_1+\tau\lV\bm{\A}\bm{v}\rV_2,~  \forall \bm{v}\in \mathbb{R}^N
       \end{equation*}
       holds for any subset $\mathcal{S}\subset [N]$, $\lv\mathcal{S}\rv\le s$. $\overline{\mathcal{S}}$ denotes the complement of $\mathcal{S}$ in $[N]$.
\end{definition}
\noindent Now we give the proof of theorem~\ref{theory:l1norm}.
\begin{proof}
       Using the notation in~\eqref{problem:CS}, the problem~\eqref{problem:p1} can be equivalently written as
       \begin{equation}\label{problem:l1reform}
              \mathop{\mathrm{minimize}}\limits_{\bm{\gamma}\in\mathbb{R}^N}~\lV\bm{\gamma}\rV_1,~\text{s.t.}~ \lV \bm{u}-\bm{\A}\bm{\gamma}\rV_2\le \eta.
       \end{equation}
       Now define a centered version of $\bm{\A}$ denoted by $\mathring{\bm{\A}}$, where  $\mathring{\bm{\A}}_k\in \mathbb{C}^{D(D-1)}$ (i.e., the $k$-th column of $\mathring{\bm{\A}}$) is formed by the vectorization of non-diagonal elements of $\bm{a}_k\bm{a}_k^*$, and thus $\mathring{\bm{\A}}\in \mathbb{C}^{D(D-1)\times N}$. According to {\cite[Theorem 2]{fengler2021non}} and {\cite[Theorem 10]{fengler2021non}}, we know there exists universal constants $\tilde{c}_2,\tilde{c}_3$, such that as long as
       \begin{equation*}
              2K\le\frac{\tilde{c}_2D(D-1)}{\log^2(eN/D)},
       \end{equation*}
       then with probability at least $1-e^{\tilde{c}_3D}$ it holds that
       \begin{equation*}
              \lV\bm{v}\rV_1\le 0.68\lV\bm{v}_{\overline{\mathcal{S}}}\rV_1+\frac{3}{D}\lV\mathring{\bm{\A}}\bm{v}\rV_2,~  \forall \bm{v}\in \mathbb{R}^N
       \end{equation*}
       for any subset $\mathcal{S}\subset [N]$, $\lv\mathcal{S}\rv\le K$. It means  $\mathring{\bm{\A}}$ satisfies the robust $\ell_2$-NSP of order $K$ under the condition above. Further, by the definition of  $\mathring{\bm{\A}}$, we know $\lV\bm{\A}\bm{v}\rV_2\ge \lV\mathring{\bm{\A}}\bm{v}\rV_2$, thus $\bm{\A}$ also satisfies the robust $\ell_2$-NSP of order $K$ with constants $0.68$ and $\frac{3}{D}$. Therefore, by \cite[Theorem 4.19]{foucart2013} we have
       \begin{align*}
              \lV\hat{\bm{\gamma}}-\bm{\gamma}^{\natural}\rV_{1}
              \le \frac{12/D}{1-0.68}\eta
              \le\frac{37.5\left(\lV\bm{\gamma}^{\natural}\rV_1+\sigma^2\right)}{\sqrt{M}}\sqrt{c\log \left(\frac{(eD)^2}{\varepsilon}\right)},
       \end{align*}
       where the first inequality follows from $\sigma\left(\bm{\gamma}^\natural\right)_1=0$ as $\lV\bm{\gamma}^\natural\rV_0\le K$, the last inequality follows from Lemma~\ref{bound:Delta}. Recall that $c$ is a universal constant. By setting $c_1=37.5\sqrt{c}$, then the proof is completed.
\end{proof}

\section{Proof of Theorem~\ref{theorem:IHTconv}}\label{proof:IHTconv}
Before the proof, we recall the definition of restricted isometry constant and a lemma for the convergence guarantees of IHT algorithm in ~\cite{foucart2013}.
\begin{definition}\label{def:RIPconstant}
       The $s$-th restricted isometry constant (denoted $\delta_{s}$) of the matrix $\bm{B}\in \mathbb{C}^{D^2\times N}$ is the smallest $\tilde{\delta}>0$ such that
       \begin{equation*}
              (1-\tilde{\delta})\lV\bm{v}\rV_2^2\le \lV\bm{B}\bm{v}\rV_2^2\le (1-\tilde{\delta})\lV\bm{v}\rV_2^2
       \end{equation*}
       for all $s$-sparse vectors $\bm{v}\in\mathbb{C}^N$.
\end{definition}
\begin{lemma}[{[\citenum{foucart2013}, Theorem 6.18]}]\label{lemma:thresholdingalgconv}
       Suppose that the $3s$-th restricted isometry constant of the matrix $\frac{\bm{\A}}{D}\in \mathbb{C}^{D^2\times N}$ satisfies  $\delta_{3s}\le\frac{1}{\sqrt{3}}$. Then, for any $s$-sparse $\bm{\gamma}^{\natural}$ and $\bm{\delta}$ in \eqref{problem:CS} with  $\lV\bm{\delta}\rV_2\le \eta$, the sequence $\{\bm{\gamma}_{t}\}_{t\ge 0}$ defined by IHT in \eqref{iter:IHT} with some constant step size $\alpha$ satisfies
       \begin{equation*}
              \lV\bm{\gamma}_{t+1}-\bm{\gamma}^\natural\rV_2\le \sqrt{3}\delta_{3s}\lV\bm{\gamma}_{t}-\bm{\gamma}^\natural\rV_2+\frac{2.18}{1-\sqrt{3}\delta_{3s}}\eta,~\forall t\ge 0.
       \end{equation*}
\end{lemma}
Now we give the proof of theorem~\ref{theorem:IHTconv}.
\begin{proof}
       Let $\hat{\mu}$ be the mutual incoherence of $\bm{\A}$. By the proof in Lemma~\ref{lemma:MIPinA}, we know $\lV\bm{\A}_k\rV_2=D$, and thus $\frac{\bm{\A}}{D}$ has $\ell_2$-normalized columns. By {\cite[Theorem 5.3]{foucart2013}} and {\cite[Eq. (5.2)]{foucart2013}} we have
       \begin{equation}\label{eq:RIPAd}
              (1-(3K-1)\hat{\mu})\lV\bm{v}\rV_2^2\le \lV\frac{\bm{\A}}{D}\bm{v}\rV_2^2\le (1+(3K-1)\hat{\mu})\lV\bm{v}\rV_2^2
       \end{equation}
       for all vectors in $\{\bm{v}\in\mathbb{C}^N: \lV\bm{v}\rV_0\le 3K\}$.  As long as the mutual incoherence of the matrix $\bm{A}=[\bm{a}_1,\bm{a}_2,\cdots,\bm{a}_N]$ satisfies $\mu\le\frac{1}{3^{1/4}\sqrt{3K-1}}$, then by Lemma~\ref{lemma:MIPinA}, we know the mutual incoherence of $\frac{\bm{\A}}{D}$ satisfies $\hat{\mu}\le\frac{1}{\sqrt{3}(3K-1)}$. Therefore, by~\eqref{eq:RIPAd} we have
       \begin{equation*}
              (1-\frac{1}{\sqrt{3}})\lV\bm{v}\rV_2^2\le \lV\frac{\bm{\A}}{D}\bm{v}\rV_2^2\le (1+\frac{1}{\sqrt{3}}\lV\bm{v}\rV_2^2
       \end{equation*}
       for all vectors in $\{\bm{v}\in\mathbb{C}^N: \lV\bm{v}\rV_0\le 3K\}$. By the definition of the restricted isometry constant in Definition~\ref{def:RIPconstant}, we have  $ \delta_{3K}\le \frac{1}{\sqrt{3}} $.
       Recall that $\bm{\gamma}^\natural$ is $K$-sparse. By Lemma~\ref{lemma:thresholdingalgconv}, we then complete the proof.
\end{proof}

\section{Proof of Theorem~\ref{theorem:Yrecovery}}\label{proof:theoremYrec}
\begin{proof}
       By simple reformulation, we have
       \begin{align}\label{eq:reformYk}
              \Y_k= & ~\sum_{1\le i,j \le D}\overline{a}_{k,i}a_{k,j}\left(\Sigma_y(i,j)+\Delta(i,j)\right)\nonumber                                                                                                            \\
              =     & ~\sum_{1\le i,j \le D}\overline{a}_{k,i}a_{k,j}\left(\sum_{\ell\in S^{\natural}}a_{\ell,i}\overline{a}_{\ell,j}\gamma^{\natural}_{\ell}\right) +\sum_{1\le i,j \le D}\overline{a}_{k,i}a_{k,j}\Delta(i,j)
              \nonumber                                                                                                                                                                                                         \\
              =     & ~\sum_{\ell\in S^{\natural}}\gamma^{\natural}_{\ell}\lv\sum_{i=1}^D\overline{a}_{k,i} a_{\ell,i}\rv^2+\A_k^*\bm{\delta}.
       \end{align}
       Recall that in \eqref{problem:l0recovery} we have defined
       $\eta=\frac{D\left(\sum_{\ell\in S^{\natural}}\gamma^{\natural}_{\ell}+\sigma^2\right)}{\sqrt{M}}\sqrt{c\log \left(\frac{(eD)^2}{\varepsilon}\right)}$.
       Then, for the later term in \eqref{eq:reformYk}, we have
       \begin{align*}
              \lv \A_k^*\bm{\delta}\rv\le \lV \A_k\rV_2 \lV \bm{\delta}\rV_2 \le D\eta,~\forall k\in [N]
       \end{align*}
       with probability at least $1-\varepsilon$ provided $M\ge c\log \left(\frac{(eD)^2}{\varepsilon}\right)$. The first inequality follows from Cauchy-Schwarz inequality. The last inequality follows from Lemma~\ref{bound:Delta} and the fact that $\lV \A_k\rV_2 =D$.
       Therefore, $\forall k_1\in S^{\natural}$ we have,
       \begin{align}\label{ineq:lowboundYk}
              \Y_{k_1}
              =\gamma^{\natural}_{k_1}\lv\sum_{i=1}^D\lv a_{k_1,i}\rv^2\rv^2+\sum_{\ell\in S^{\natural},\ell\neq k_1}\gamma^{\natural}_{\ell}\lv\sum_{i=1}^D\overline{a}_{k_1,i} a_{\ell,i}\rv^2+\A_{k_1}^*\bm{\delta}
              \ge \gamma^{\natural}_{k_1}D^2- D\eta.
       \end{align}
       Also, by \eqref{eq:reformYk} and Lemma~\ref{lemma:conAk} we have
       \begin{align}\label{ineq:upboundYk}
              \Y_{k_2}\le \sum_{\ell\in S^{\natural}}32\gamma^{\natural}_{\ell}D\log (DN) + D\eta,~\forall k_2\notin S^{\natural}
       \end{align}
       holds with probability at least $1-\frac{c_0}{D\sqrt{D}}$. Then, we must have
       \begin{equation*}
              \Y_{k_1}>  \Y_{k_2} ~ \text{for any}~k_1\in S^{\natural}~\text{and any}~k_2\notin S^{\natural},
       \end{equation*}
       as long as
       \begin{align}\label{cond:delta}
              0< \frac{1}{2}\gamma_k^{\natural}-\sum_{\ell\in S^{\natural}}\gamma^{\natural}_{\ell}\frac{16\log (DN)}{D}
              -\frac{\eta}{D}, \quad \forall k\in[N].
       \end{align}
       To guarantee \eqref{cond:delta}, we let
       \begin{equation}\label{ineq:condDM}
              D\ge \frac{\tilde{c}_4\sum_{\ell\in S^{\natural}}\gamma^{\natural}_{\ell}\log(DN)}{\gamma^{\natural}_{\min}}
              \quad\mathrm{and}\quad
              M\ge c_4\left(\frac{\sum_{\ell\in S^{\natural}}\gamma^{\natural}_{\ell}+\sigma^2}{\gamma^{\natural}_{\min}}\right)^2\log \left(\frac{(eD)^2}{\varepsilon}\right),
       \end{equation}
       for some universal constant $\tilde{c}_4$ and $c_4$.  Therefore, \eqref{cond:delta} is satisfied under the conditions in \eqref{ineq:condDM}, and thus $S=S^{\natural}$. By noticing that
       $\frac{\sum_{\ell\in S^{\natural}}\gamma^{\natural}_{\ell}+\sigma^2}{\gamma^{\natural}_{\min}} \le \frac{\gamma^{\natural}_{\max}K}{\gamma^{\natural}_{\min}}+\frac{\sigma^2}{\gamma^{\natural}_{\min}}$,
       we complete the proof.
\end{proof}

\section{Proof of Theorem~\ref{theorem:rhoYrecovery}}\label{proof:rhotheoremYrec}
\begin{proof}
       Similar to \eqref{ineq:lowboundYk}, by \eqref{eq:reformYk}, Lemma~\ref{lemma:conAk} and triangle inequality we have
       \begin{align*}
              \Y_{k_1} \le \gamma^{\natural}_{k_1}D^2+\sum_{\ell\in S^{\natural}}32\gamma^{\natural}_{\ell}D\log (DN) +D\eta,~\forall {k_1}\in S^\natural.
       \end{align*}
       Then by the above inequality together with \eqref{ineq:upboundYk},  we have
       \begin{align*}
              \frac{1}{N} \sum_{k=1}^N\Y_k \le \frac{D^2}{N}\sum_{k=1}^N\gamma^{\natural}_{k}+\sum_{\ell\in S^{\natural}}32\gamma^{\natural}_{\ell}D\log (DN) +D\eta.
       \end{align*}
       Then, using the same argument to the proof of Theorem~\ref{theorem:Yrecovery} in Appendix~\ref{proof:theoremYrec}, we have
       \begin{equation*}
              \Y_k> \frac{1}{N} \sum_{i=1}^N\Y_i,\quad \forall k\in S^\natural.
       \end{equation*}
       provided that $M$ satisfies \eqref{ineq:condDM} and
       $K \le \min\{\frac{c_6 D}{c_\gamma\log(DN)},\frac{c_7 N}{c_\gamma}\}$,
       for some sufficiently small universal constants $c_6,c_7$ and $c_\gamma=\frac{\gamma_{\max}}{\gamma_{\min}}$. Thus, by the definition of $S^0$ we have $S^\natural\subset S^0$.
\end{proof}


\bibliographystyle{IEEEtran}
\bibliography{actdec}
\end{document}